\documentclass{article}

\usepackage[final,nonatbib]{neurips_2019}

\usepackage[utf8]{inputenc} 
\usepackage[T1]{fontenc}    
\usepackage{hyperref}       
\usepackage{url}            
\usepackage{booktabs}       
\usepackage{amsfonts}       
\usepackage{nicefrac}       
\usepackage{microtype}      

\usepackage{amsmath}
\usepackage{graphicx}
\usepackage{adjustbox}
\graphicspath{{figs/}}
\usepackage{xargs}
\usepackage[colorinlistoftodos]{todonotes}
\usepackage{amsthm}
\usepackage{amssymb}
\usepackage{mathtools}
\usepackage{bbm}
\usepackage{ dsfont }
\usepackage{subcaption}

\usepackage{xcolor}
\usepackage[normalem]{ulem}

\usepackage{tikz}
\usetikzlibrary{arrows,automata}
\usepackage{pgfplots}
\usepackage{array}
\usepackage{ifthen}
\usepackage{enumitem}

\usepackage{theoremref}

\newcommand{\Setspt}[1][\text{ }]{\mathcal{T}_{#1}} 
\newcommandx{\SetsptE}[2][1=\bar{E},2=]{\Setspt[#2]^{#1}}
\newcommandx{\OpT}[3][1=\bar{E},2=\hat{E}, 3=t]{{#3}_{#1}^{#2}}
\newcommandx{\SEd}[3][1=\bar{E},2=t,3=]{\mathds{E}^{#3}_{#2,#1}}

\newcommandx{\EfR}[2][1=u,2=v]{r_{#1#2}^{\operatorname{eff}}}
\newcommandx{\Forest}[3][1=u,2=v,3=]{\ifthenelse{\equal{#3}{}}{\mathcal{F}_{#1}^{#2}}{\mathcal{F}_{#1,#3}^{#2}}}

\theoremstyle{definition}
\newtheorem{definition}{Definition}[section]
\newtheorem{theorem}{Theorem}[section]
\theoremstyle{remark}

\newtheorem{lemma}[theorem]{Lemma}

\title{Probabilistic Watershed: \\Sampling all spanning forests \\for seeded segmentation and semi-supervised learning}

\author{%
	Enrique Fita Sanmartín,
	  \hspace{0.5cm} Sebastian Damrich, \hspace{0.5cm} Fred A. Hamprecht\\
	HCI/IWR at Heidelberg University, 69115 Heidelberg, Germany\\
	\texttt{\{fita@stud, sebastian.damrich@iwr, fred.hamprecht@iwr\}.uni-heidelberg.de}
}
\begin{document}

\maketitle

\begin{abstract}
The seeded Watershed algorithm / minimax semi-supervised learning on a graph computes a minimum spanning forest which connects every pixel / unlabeled node to a seed / labeled node. We propose instead to consider \textit{all possible} spanning forests and calculate, for every node, the probability  of sampling a forest connecting a certain seed with that node. We dub this approach "Probabilistic Watershed". Leo Grady (2006) already noted its equivalence to the Random Walker / Harmonic energy minimization. We here give a simpler proof of this equivalence and establish the computational feasibility of the Probabilistic Watershed with Kirchhoff's matrix tree theorem. Furthermore, we show a new connection between the Random Walker probabilities and the triangle inequality of the effective resistance. Finally, we derive a new and intuitive interpretation of the Power Watershed.
\end{abstract}
\section{Introduction}
\label{sec:Introduction}
Seeded segmentation in computer vision and graph-based semi-supervised machine learning are essentially the same problem. In both, a popular paradigm is the following: given many unlabeled pixels / nodes in a graph as well as a few seeds / labeled nodes, compute a distance from a given query pixel / node to  all of the seeds,  and assign the query to a class based on the shortest distance. 

There is obviously a large selection of distances to choose from, and popular choices include: 
\textit{i)} the shortest path distance (e.g.~\cite{criminisi2008geos}), 
\textit{ii)} the commute distance (e.g.~\cite{Zhu2003,zhou2004learning,belkin2006manifold,Grady2006}) or 
\textit{iii)} the bottleneck shortest path distance (e.g.~\cite{kim2014label,najman2019}). Thanks to its matroid property, the latter can be computed very efficiently -- a greedy algorithm finds the global optimum -- and is thus widely studied and used in different fields under names including widest, minimax, maximum capacity, topographic and watershed path distance. In computer vision, the corresponding algorithm known as ``Watershed'' is popular in seeded segmentation not only because it is so efficient \cite{Chazelle2000} but also because it works well in a broad range of problems \cite{Learned_Watershed,DeepWatershed}, is well understood theoretically \cite{Cousty2009, Allene2007}, and unlike Markov Random Fields induces no shrinkage bias \cite{beier2017multicut}. 
Even though the Watershed's optimization problem can be solved efficiently, it is combinatorial in nature. One consequence is the ``winner-takes-all'' characteristic of its solutions: a pixel or node is always unequivocally assigned to a single seed. Given suitable graph edge-weights, this solution is often but not always correct, see \figurename s \ref{fig:summary_ProbWS} and \ref{fig:example_performance}\footnote{which were produced with the code at \url{https://github.com/hci-unihd/Probabilistic_Watershed}}. 

\tikzset{
	dot/.style 2 args={fill, circle, inner sep=0pt, label={#1:\scriptsize #2}},	
	fulldot/.style 2 args={circle,draw,minimum size=0.3cm,inner sep=0pt, label={#1:\scriptsize #2}},
	graph node/.style={circle,draw,minimum size=0.25cm,inner sep=0pt]},
	graph node2/.style={circle,draw,minimum size=0.1cm,inner sep=0pt]},
	invisible/.style={circle,minimum size=0.001cm,inner sep=0pt]},
	cut node/.style={minimum size=0.115cm,inner sep=0pt]},
	cut node2/.style={minimum size=0.075cm,inner sep=0pt]},
	main node/.style={circle,draw,minimum size=0.5cm,inner sep=0pt]},
	punt/.style={circle,draw,minimum size=0.03cm,inner sep=0pt]},
}
%
%
%


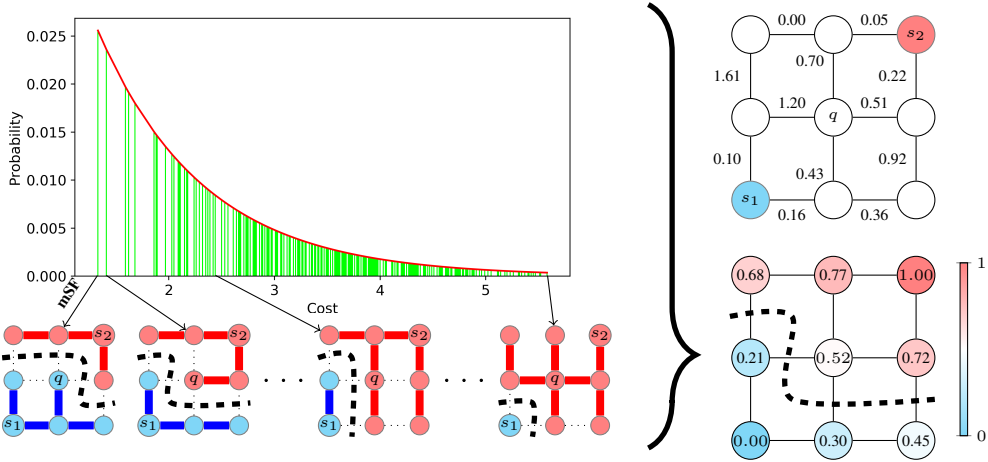
\begin{figure*}

        \begin{tikzpicture}
        \node[] (plot) at (-5.2,1.7) {\includegraphics[width=8cm,height=4cm,trim={0.1cm 0cm 1cm 1cm},clip]{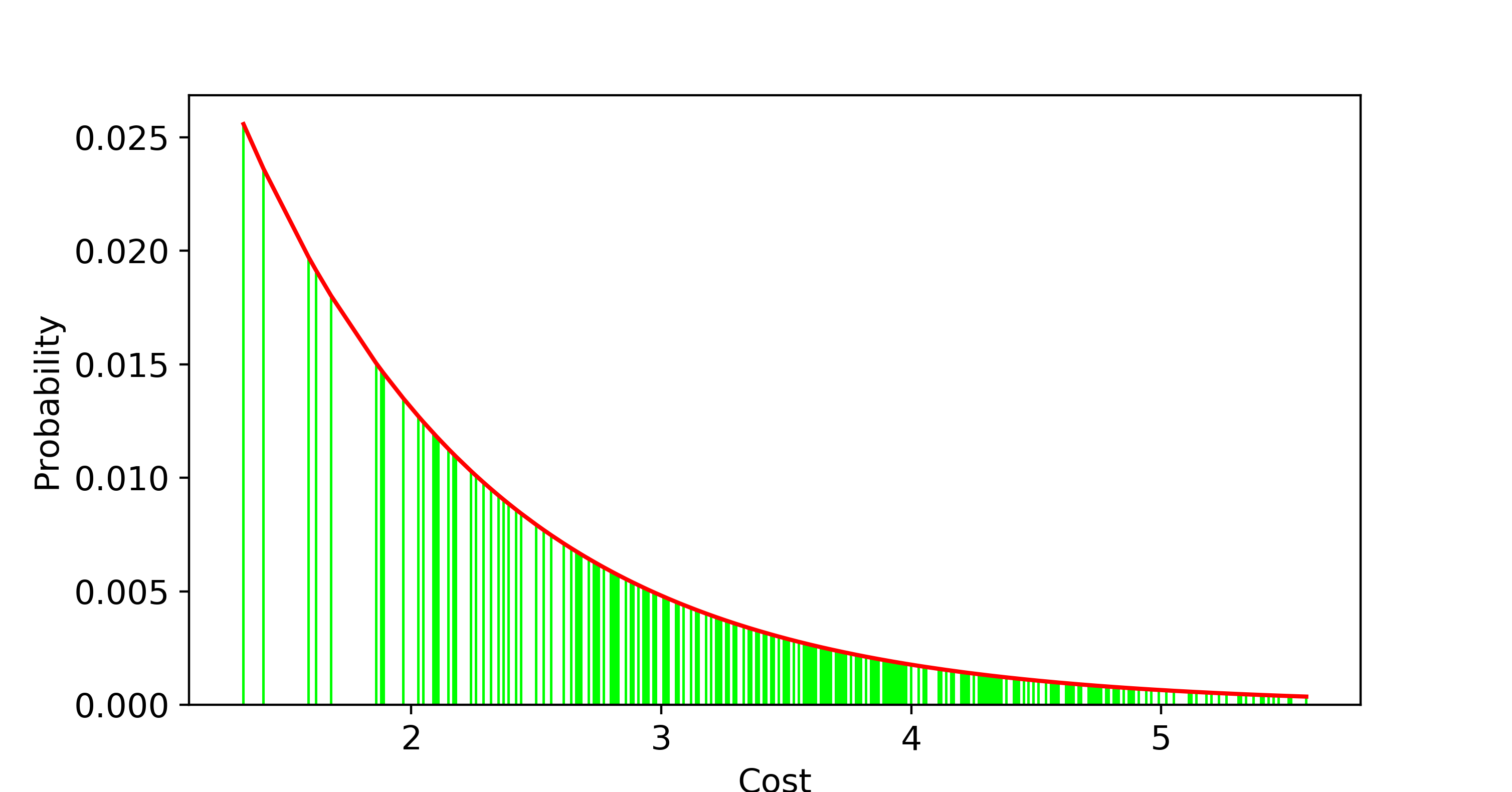}};
        \draw [line width=2pt,decorate,decoration={brace,amplitude=18pt,raise=4pt,aspect=0.79},yshift=0pt]
        (-0.7,3.8) -- (-0.7,-2.1) node [black,midway,xshift=0.8cm] {};
        \def\scale{.6 }
        \def\scaled{1.2 }
        \def\tx{-7.5}				
        \def\tyd{-1.5}
        \node[graph node,opacity=.5,black,fill=cyan,text opacity=1] (g20) at (0+\scaled*\tx,0+\scaled*\tyd) {\tiny{$s_1$}};
        \node[graph node,opacity=.5,fill=cyan,] (g21) at (\scale*1+\scaled*\tx,0+\scaled*\tyd) {};
        \node[graph node,opacity=.5,fill=cyan,] (g22) at (\scale*2+\scaled*\tx,0+\scaled*\tyd) {};
        \node[graph node,opacity=.5,fill=cyan,] (g23) at (0+\scaled*\tx,\scale*1+\scaled*\tyd) {};
        \node[graph node,opacity=.5,fill=cyan,text opacity=1] (g24) at (\scale*1+\scaled*\tx,\scale*1+\scaled*\tyd) {\tiny{$q$}};
        \node[graph node,opacity=.5,fill=red,] (g25) at (\scale*2+\scaled*\tx,\scale*1+\scaled*\tyd) {};
        \node[graph node,opacity=.5,fill=red,] (g26) at (0+\scaled*\tx,\scale*2+\scaled*\tyd) {};
        \node[graph node,opacity=.5,fill=red,] (g27) at (\scale*1+\scaled*\tx,\scale*2+\scaled*\tyd) {};
        \node[graph node,opacity=.5,black,fill=red,text opacity=1] (g28) at (\scale*2+\scaled*\tx,\scale*2+\scaled*\tyd) {\tiny{$s_2$}};

        \path[dotted,draw]
		
		(g22) edge node[left] {} (g25)

		(g23) edge node[above] {} (g24)
		
		(g23) edge node[left] {} (g26)
		
		(g24) edge node[above] {} (g25)
		
		(g24) edge node[above left] {} (g27);

		\path[-,draw,blue,line width=3pt]
		(g20) edge node[below] {} (g21)
		
		(g20) edge node[left] {} (g23)
		
		(g21) edge node[below] {} (g22)
				
		(g21) edge node[below left] {} (g24);
		
		\path[-,draw,red,line width=3pt]
		(g26) edge node[above] {} (g27)
		
		(g27) edge node[above] {} (g28)
		
		(g25) edge node[left] {} (g28);
		
		\node[cut node2] (cl0) at (0+\scaled*\tx-0.25*\scale,\scale*0.5+\scale*1+\scaled*\tyd) {};
		\node[cut node2] (cl1) at (0+\scaled*\tx+1.5*\scale,\scale*0.5+\scale*1+\scaled*\tyd) {};
		
		\node[cut node2] (cl2) at (0+\scaled*\tx+1.5*\scale,\scale*0.5+0*\scale+\scaled*\tyd) {};
		
        \node[cut node2] (cl3) at (0+\scaled*\tx+2.25*\scale,\scale*0.5+0*\scale+\scaled*\tyd) {};

        \draw [dashed,line width=2pt,xshift=4cm]  plot[smooth] coordinates { (cl0) (cl1) (cl2) (cl3)};

        
        \def\tx{-6}				
        \def\ty{\tyd}
        \node[graph node,opacity=.5,black,fill=cyan,text opacity=1] (g10) at (0+\scaled*\tx,0+\scaled*\ty) {\tiny{$s_1$}};
        \node[graph node,opacity=.5,fill=cyan,] (g11) at (\scale*1+\scaled*\tx,0+\scaled*\ty) {};
        \node[graph node,opacity=.5,fill=cyan,] (g12) at (\scale*2+\scaled*\tx,0+\scaled*\ty) {};
        \node[graph node,opacity=.5,fill=cyan,] (g13) at (0+\scaled*\tx,\scale*1+\scaled*\ty) {};
        \node[graph node,opacity=.5,fill=red,text opacity=1] (g14) at (\scale*1+\scaled*\tx,\scale*1+\scaled*\ty) {\tiny{$q$}};
        \node[graph node,opacity=.5,fill=red,] (g15) at (\scale*2+\scaled*\tx,\scale*1+\scaled*\ty) {};
        \node[graph node,opacity=.5,fill=red,] (g16) at (0+\scaled*\tx,\scale*2+\scaled*\ty) {};
        \node[graph node,opacity=.5,fill=red,] (g17) at (\scale*1+\scaled*\tx,\scale*2+\scaled*\ty) {};
        \node[graph node,opacity=.5,black,fill=red,text opacity=1] (g18) at (\scale*2+\scaled*\tx,\scale*2+\scaled*\ty) {\tiny{$s_2$}};

        \node[punt,fill=black] (punt1) at (\scale*2+\scaled*\tx+0.36,\scale*1+\scaled*\ty) {};
        \node[punt,fill=black]  (punt2) at (\scale*2+\scaled*\tx+0.59,\scale*1+\scaled*\ty){};
        \node[punt,fill=black] (punt3) at (\scale*2+\scaled*\tx+0.82,\scale*1+\scaled*\ty){};

        \path[dotted,draw]
		
		(g11) edge node[below left] {} (g14)
		
		(g12) edge node[left] {} (g15)
		
		(g13) edge node[above] {} (g14)
		
		(g13) edge node[left] {} (g16)
		
		(g14) edge node[above left] {} (g17);

		\path[-,draw,blue,line width=3pt]
		(g10) edge node[below] {} (g11)
		
		(g11) edge node[below] {} (g12)
				
		(g10) edge node[left] {} (g13);
		
		\path[-,draw,red,line width=3pt]
		
		(g14) edge node[above] {} (g15)
		
		(g15) edge node[left] {} (g18)
		
		(g16) edge node[above] {} (g17)
		
		(g17) edge node[above] {} (g18);
		
		\node[cut node2] (cm0) at (0+\scaled*\tx-0.25*\scale,\scale*0.5+\scale*1+\scaled*\tyd) {};
		\node[cut node2] (cm1) at (0+\scaled*\tx+0.5*\scale,\scale*0.5+\scale*1+\scaled*\tyd) {};
		
		\node[cut node2] (cm2) at (0+\scaled*\tx+0.5*\scale,\scale*0.5+0*\scale+\scaled*\tyd) {};
		
        \node[cut node2] (cm3) at (0+\scaled*\tx+2.25*\scale,\scale*0.5+0*\scale+\scaled*\tyd) {};

        \draw [dashed,line width=2pt,xshift=4cm]  plot[smooth] coordinates { (cm0) (cm1) (cm2) (cm3)};

        \def\tx{-4 }				
        \def\tyt{\tyd }
        \node[graph node,opacity=.5,black,fill=cyan,text opacity=1] (g30) at (0+\scaled*\tx,0+\scaled*\tyt) {\tiny{$s_1$}};
        \node[graph node,opacity=.5,fill=red,] (g31) at (\scale*1+\scaled*\tx,0+\scaled*\tyt) {};
        \node[graph node,opacity=.5,fill=red,] (g32) at (\scale*2+\scaled*\tx,0+\scaled*\tyt) {};
        \node[graph node,opacity=.5,fill=cyan,] (g33) at (0+\scaled*\tx,\scale*1+\scaled*\tyt) {};
        \node[graph node,opacity=.5,fill=red,text opacity=1] (g34) at (\scale*1+\scaled*\tx,\scale*1+\scaled*\tyt) {\tiny{$q$}};
        \node[graph node,opacity=.5,fill=red,] (g35) at (\scale*2+\scaled*\tx,\scale*1+\scaled*\tyt) {};
        \node[graph node,opacity=.5,fill=red,] (g36) at (0+\scaled*\tx,\scale*2+\scaled*\tyt) {};
        \node[graph node,opacity=.5,fill=red,] (g37) at (\scale*1+\scaled*\tx,\scale*2+\scaled*\tyt) {};
        \node[graph node,opacity=.5,black,fill=red,text opacity=1] (g38) at (\scale*2+\scaled*\tx,\scale*2+\scaled*\tyt) {\tiny{$s_2$}};
        
        \node[punt,fill=black] (punt4) at (\scale*2+\scaled*\tx+0.36,\scale*1+\scaled*\ty) {};
        \node[punt,fill=black]  (punt5) at (\scale*2+\scaled*\tx+0.59,\scale*1+\scaled*\ty){};
        \node[punt,fill=black] (punt6) at (\scale*2+\scaled*\tx+0.82,\scale*1+\scaled*\ty){};
        
       \path[dotted,draw]
       
		(g30) edge node[below] {} (g31)
		
		(g31) edge node[below] {} (g32)
		
		(g33) edge node[above] {} (g34)
		
		(g33) edge node[left] {} (g36)
		
		(g34) edge node[above] {} (g35);

		\path[-,draw,blue,line width=3pt]
				
		(g30) edge node[left] {} (g33);
		
		\path[-,draw,red,line width=3pt]
		
		(g31) edge node[below left] {} (g34)
		
		(g32) edge node[left] {} (g35)
		
		(g34) edge node[above left] {} (g37)
		
		(g35) edge node[left] {} (g38)
		
		(g36) edge node[above] {} (g37)
		
		(g37) edge node[above] {} (g38);

		\node[cut node2] (cr0) at (0+\scaled*\tx-0.25*\scale,\scale*0.5+\scale*1+\scaled*\tyd) {};
		\node[cut node2] (cr1) at (0+\scaled*\tx+0.5*\scale,\scale*0.5+\scale*1+\scaled*\tyd) {};
		
		\node[cut node2] (cr2) at (0+\scaled*\tx+0.5*\scale,-\scale*0.25+0*\scale+\scaled*\tyd) {};

        \draw [dashed,line width=2pt,xshift=4cm]  plot[smooth] coordinates { (cr0) (cr1) (cr2)};

        \def\tx{-2 }				
        \node[graph node,opacity=.5,black,fill=cyan,text opacity=1] (g40) at (0+\scaled*\tx,0+\scaled*\tyt) {\tiny{$s_1$}};
        \node[graph node,opacity=.5,fill=red,] (g41) at (\scale*1+\scaled*\tx,0+\scaled*\tyt) {};
        \node[graph node,opacity=.5,fill=red,] (g42) at (\scale*2+\scaled*\tx,0+\scaled*\tyt) {};
        \node[graph node,opacity=.5,fill=red,] (g43) at (0+\scaled*\tx,\scale*1+\scaled*\tyt) {};
        \node[graph node,opacity=.5,fill=red,text opacity=1] (g44) at (\scale*1+\scaled*\tx,\scale*1+\scaled*\tyt) {\tiny{$q$}};
        \node[graph node,opacity=.5,fill=red,] (g45) at (\scale*2+\scaled*\tx,\scale*1+\scaled*\tyt) {};
        \node[graph node,opacity=.5,fill=red,] (g46) at (0+\scaled*\tx,\scale*2+\scaled*\tyt) {};
        \node[graph node,opacity=.5,fill=red,] (g47) at (\scale*1+\scaled*\tx,\scale*2+\scaled*\tyt) {};
        \node[graph node,opacity=.5,black,fill=red,text opacity=1] (g48) at (\scale*2+\scaled*\tx,\scale*2+\scaled*\tyt) {\tiny{$s_2$}};

       \path[dotted,draw]
       
		(g40) edge node[below] {} (g41)
		
		(g40) edge node[left] {} (g43)
		
		(g41) edge node[below] {} (g42);

		\path[-,draw,red,line width=3pt]

		(g41) edge node[below left] {} (g44)
		
		(g42) edge node[left] {} (g45)
		
		(g43) edge node[above] {} (g44)
		
		(g43) edge node[left] {} (g46)
		
		(g44) edge node[above] {} (g45)
		
		(g44) edge node[above left] {} (g47)
		
		(g45) edge node[left] {} (g48);

        \node[cut node2] (cw0) at (0+\scaled*\tx-0.25*\scale,\scale*0.5+\scale*0+\scaled*\tyd) {};
		\node[cut node2] (cw1) at (0+\scaled*\tx+0.5*\scale,\scale*0.5+\scale*0+\scaled*\tyd) {};
		
		\node[cut node2] (cw2) at (0+\scaled*\tx+0.5*\scale,-\scale*0.25+0*\scale+\scaled*\tyd) {};

        \draw [dashed,line width=2pt,xshift=4cm]  plot[smooth] coordinates { (cw0) (cw1) (cw2)};


		\node[invisible] (extra0) at (-7.88,0.2) {};
		
		\node[invisible] (extra1) at (-7.77,0.2) {};
	
		\node[invisible] (extra2) at (-6.32,0.2) {};
		
		\node[invisible] (extra3) at (-1.9,0.2) {};
		
		\path[->,draw]
		(extra0) edge node[sloped,above] {\tiny \textbf{mSF}} (g27)
		(extra1) edge node[below] {} (g17)
		(extra2) edge node[below] {} (g36)
		(extra3) edge node[below] {} (g47);

		\def\scale{1.1}
		\def\biasx{0.8}
		\def\biasyraw{1.2}
		\node[main node,opacity=.5,black,fill=cyan,text opacity=1] (0) at (0+\biasx,0+\biasyraw) {\tiny$s_1$};
		\node[main node] (1) at (\scale*1+\biasx,0+\biasyraw) {};
		\node[main node] (2) at (\scale*2+\biasx,0+\biasyraw) {};
		\node[main node] (3) at (0+\biasx,\scale*1+\biasyraw) {};
		\node[main node] (4) at (\scale*1+\biasx,\scale*1+\biasyraw) {\tiny$q$};
		\node[main node] (5) at (\scale*2+\biasx,\scale*1+\biasyraw) {};
		\node[main node] (6) at (0+\biasx,\scale*2+\biasyraw) {};
		\node[main node] (7) at (\scale*1+\biasx,\scale*2+\biasyraw) {};
		\node[main node,opacity=.5,black,fill=red,text opacity=1] (8) at (\scale*2+\biasx,\scale*2+\biasyraw) {\tiny$s_2$};

	    \path[-,draw]
		(0) edge node[below] {\tiny0.16} (1)
		
		(0) edge node[left] {\tiny0.10} (3)
		
		(1) edge node[below] {\tiny0.36} (2)
		
		(1) edge node[below left] {\tiny0.43} (4)
		
		(2) edge node[left] {\tiny0.92} (5)
        
		(3) edge node[above] {\tiny1.20} (4)
		
		(3) edge node[left] {\tiny1.61} (6)
		
		(4) edge node[above] {\tiny0.51} (5)
		
		(4) edge node[above left] {\tiny0.70} (7)
        
		(5) edge node[left] {\tiny0.22} (8)
		
		(6) edge node[above] {\tiny0.00} (7)
		
		(7) edge node[above] {\tiny0.05} (8);

		        \def\biasy{-2.0}
				\node[main node,opacity=1,black,fill={rgb,255:red,127.5; green,214.5; blue,247},text opacity=1] (0) at (0+\biasx,0+\biasy) {\tiny$0.00$};
				\node[main node,opacity=1,black,fill={rgb,255:red,204; green,238.8; blue,251.8},text opacity=1] (1) at (\scale*1+\biasx,0+\biasy) {\tiny0.30};
				\node[main node,opacity=1,black,fill={rgb,255:red,242.25; green,250.95; blue,254.2},text opacity=1] (2) at (\scale*2+\biasx,0+\biasy) {\tiny0.45};
				
				\node[main node,opacity=1,black,fill={rgb,255:red,181.05; green,231.51; blue,250.36},text opacity=1] (3) at (0+\biasx,\scale*1+\biasy) {\tiny0.21};
				\node[main node,opacity=1,black,fill={rgb,255:red,255; green,249.9; blue,249.9},text opacity=1] (4) at (\scale*1+\biasx,\scale*1+\biasy) {\tiny$0.52$};
				\node[main node,opacity=1,black,fill={rgb,255:red,255; green, 198.9; blue, 198.9},text opacity=1] (5) at (\scale*2+\biasx,\scale*1+\biasy) {\tiny0.72};
				
				\node[main node,opacity=1,black,fill={rgb,255:red,255; green,209.1; blue,209.1},text opacity=1] (6) at (0+\biasx,\scale*2+\biasy) {\tiny0.68};
				\node[main node,opacity=1,black,fill={rgb,255:red,255; green,186.15; blue,186.15},text opacity=1] (7) at (\scale*1+\biasx,\scale*2+\biasy) {\tiny0.77};
				\node[main node,opacity=1,black,fill={rgb,255:red,255; green,127.5; blue,127.5},text opacity=1] (8) at (\scale*2+\biasx,\scale*2+\biasy) {\tiny$1.00$};

				\node[cut node] (c0) at (0+\biasx-0.25*\scale,\scale*0.5+\scale*1+\biasy) {};
				\node[cut node] (c1) at (0+\biasx+0.5*\scale,\scale*0.5+\scale*1+\biasy) {};

				\node[cut node] (c2) at (1*\scale+\biasx-0.5*\scale,\scale*0.5+\scale*0+\biasy) {};

				\node[cut node,] (c3) at (2*\scale+\biasx+0.25*\scale,\scale*0.5+\scale*0+\biasy) {};
				
				\draw [dashed,line width=2pt,xshift=4cm]  plot[smooth] coordinates { (c0) (c1) (c2) (c3)};

				\path[-,draw]
				(0) edge node[below] {} (1)
				(0) edge node[left] {} (3)
				
				(1) edge node[below] {} (2)
				(1) edge node[below left] {} (4)
				
				(2) edge node[left] {} (5)

				(3) edge node[above] {} (4)
				(3) edge node[left] {} (6)
				
				(4) edge node[above] {} (5)
				(4) edge node[above left] {} (7)

				(5) edge node[left] {} (8)
				
				(6) edge node[above] {} (7)
				
				(7) edge node[above] {} (8);

				\begin{axis}[
			hide axis,
			scale only axis,
			width=20pt,
			colormap={cyanred}{
				rgb255(0)=(127.5,214.5,247)
				rgb255(1)=(255,255,255)
				rgb255(2)=(255,127.5,127.5) 
			},
			colorbar horizontal,
			point meta min=0,
			point meta max=1,
			colorbar style={
				at={(\biasx+4.3,\biasy+2.6)},
				width=2.3cm,
				height=0.1cm,
				rotate=90,
				xtick={0,1},
			},ticklabel style={right,font=\tiny},]
		
			\addplot [draw=none] coordinates {(1,8) (1,1)};
			\vspace{0.12cm}
			\end{axis}
			
\end{tikzpicture}
\caption{The Probabilistic Watershed computes the expected seed assignment of every node for a Gibbs distribution over all exponentially many spanning forests in closed-form. It thus avoids the winner-takes-all behaviour of the Watershed. (\textbf{Top right}) Graph with edge-costs and two seeds. (\textbf{Bottom left}) The minimum spanning forest (mSF) and other, higher cost forests. The Watershed selects the mSF, which assigns the query node $q$ to seed $s_1$. Other forests of low cost might however induce different segmentations. The dashed lines indicate the cut of the segmentations. For instance, the other depicted forests connect $q$ to $s_2$.  (\textbf{Top left}) We therefore consider a Gibbs distribution over all spanning forests with respect to their cost (see equation \eqref{eq:Gibbs_dist}, $\mu=1$). Each green bar corresponds to the cost of one of the 288 possible spanning forests. (\textbf{Bottom right}) Probabilistic Watershed probabilities for assigning a node to $s_2$. Query $q$ is now assigned to $s_2$. Considering a distribution over all spanning forests gives an uncertainty measure and can yield a segmentation different from the mSF's. In contrast to the 288 forests in this toy graph, for the real-life image in \figurename{}  \ref{fig:example_performance} one would have to consider at least $10^{11847}$ spanning forests separating the 13 seeds (see appendix \ref{App:Bound_Forests}), a feat impossible without the matrix tree theorem.}
\label{fig:summary_ProbWS}
\end{figure*}

Intrigued by the value of the Watershed to many computer vision pipelines, we have sought to entropy-regularize the combinatorial problem to make it more amenable to end-to-end learning in modern pipelines. Exploiting the equivalence of Watershed segmentations to minimum cost spanning forests, we hence set out from the following question: Is it possible to compute not 
just the minimum, but all (!) possible spanning forests, and to 
compute, in closed form, the probability that a pixel of interest is 
assigned to one of the seeds? More specifically, we envisaged a Gibbs 
distribution over the exponentially many distinct forests that span an 
undirected graph with edge-costs, where each forest is assigned a 
probability that decreases with increasing sum of the edge-costs in that forest.

If computed naively, this would be an intractable problem for all but the 
smallest graphs. However, we show here that a closed-form 
solution can be found by recurring to Kirchhoff's matrix tree theorem, and is given by the solution of the Dirichlet problem associated with commute distances \cite{Zhu2003,zhou2004learning,belkin2006manifold,Grady2006}. 
Leo Grady mentioned this connection in \cite{Grady2006,gradytextbook} and based his argument on potential theory, using results from \cite{Biggs1997}. Our informal poll amongst 
experts from both computer vision and machine learning indicated that this connection has remained mostly unknown. We hence offer a completely self-contained, except for the matrix tree theorem, and hopefully simpler proof. 

In this entirely conceptual work, we 
\begin{itemize}[leftmargin=*]

\item give a proof, using elementary graph constructions and building on the matrix tree theorem, that shows how to compute analytically the probability that a graph node is assigned to a particular seed in an ensemble of Gibbs distributed spanning forests (Section \ref{sec:PWS}).
\item establish equivalence to the algorithm known as Random Walker in computer vision \cite{Grady2006} and as Laplacian Regularized Least Squares and under other names in transductive machine learning \cite{Zhu2003,zhou2004learning,belkin2006manifold}. In particular, we relate, for the first time, the probability of assigning a query node to a seed to the triangle inequality of the effective resistance between seeds and query  (Section \ref{sec:Relation_Random_Walk}).
\item give a new interpretation of the so-called Power Watershed \cite{Couprie2011} (Section \ref{sec:Pow_Watershed}).
\end{itemize}

\begin{figure}
	\noindent\makebox[\linewidth][c]{
	\begin{subfigure}[t]{.24\textwidth}
		\centering
		\includegraphics[width=1\linewidth,trim={1cm 1cm 1cm 1cm},clip]{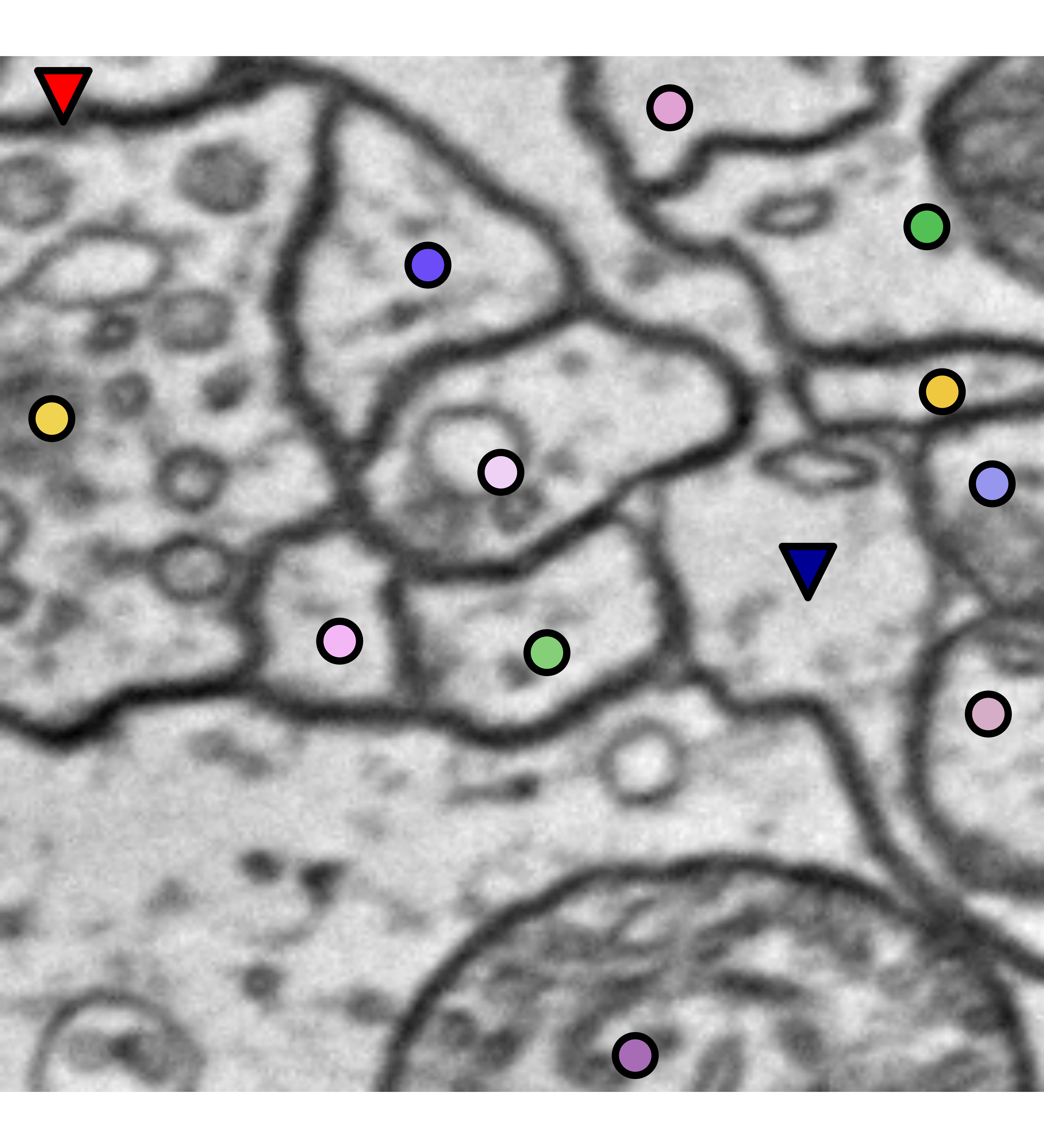}
		\caption{Image with seeds}
		\label{sfig1:example_performance}
	\end{subfigure}
	\begin{subfigure}[t]{.24\textwidth}
		\centering
		\includegraphics[width=1\linewidth,trim={1cm 1cm 1cm 1cm},clip]{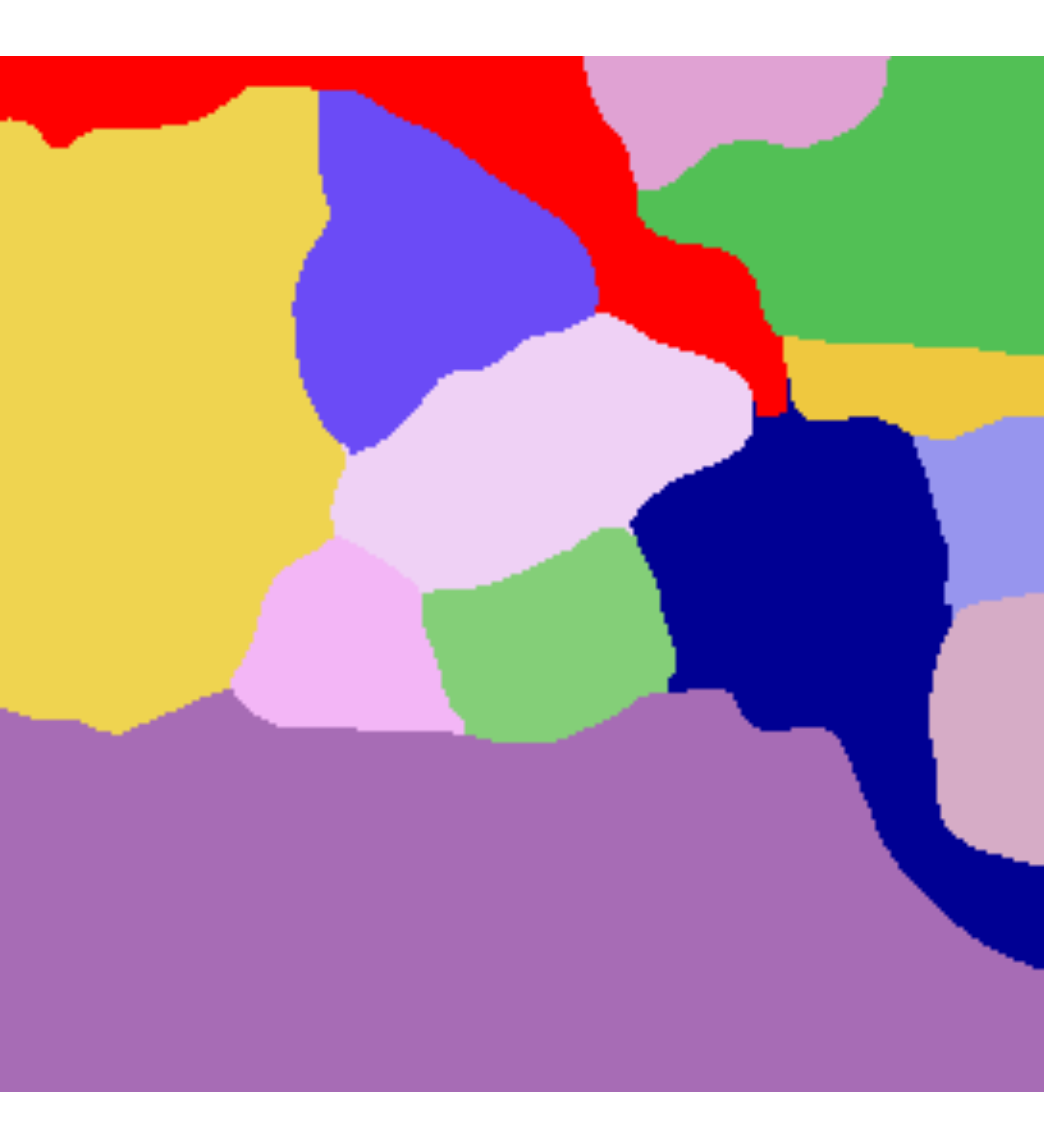}
		\caption{Watershed}
		\label{sfig2:example_performance}
	\end{subfigure}
	\begin{subfigure}[t]{.24\textwidth}
		\centering
		\includegraphics[width=1\linewidth,trim={1cm 1cm 1cm 1cm},clip]{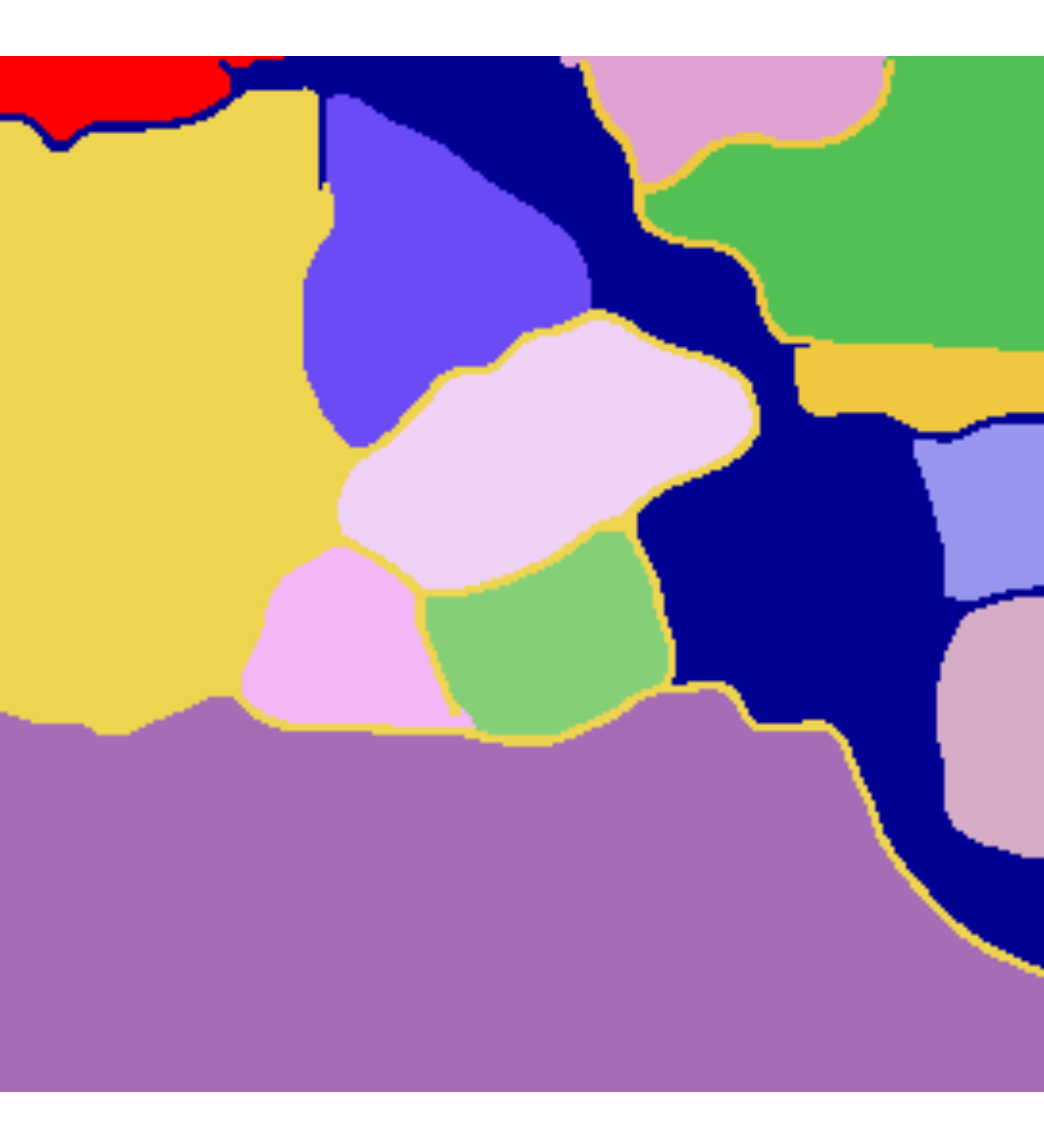}
		\caption{\centering Probabilistic Watershed}
		\label{sfig3:example_performance}
	\end{subfigure}
	\begin{subfigure}[t]{.24\textwidth}
		\centering
		\includegraphics[width=1\linewidth,trim={1cm 1cm 1cm 1cm},clip]{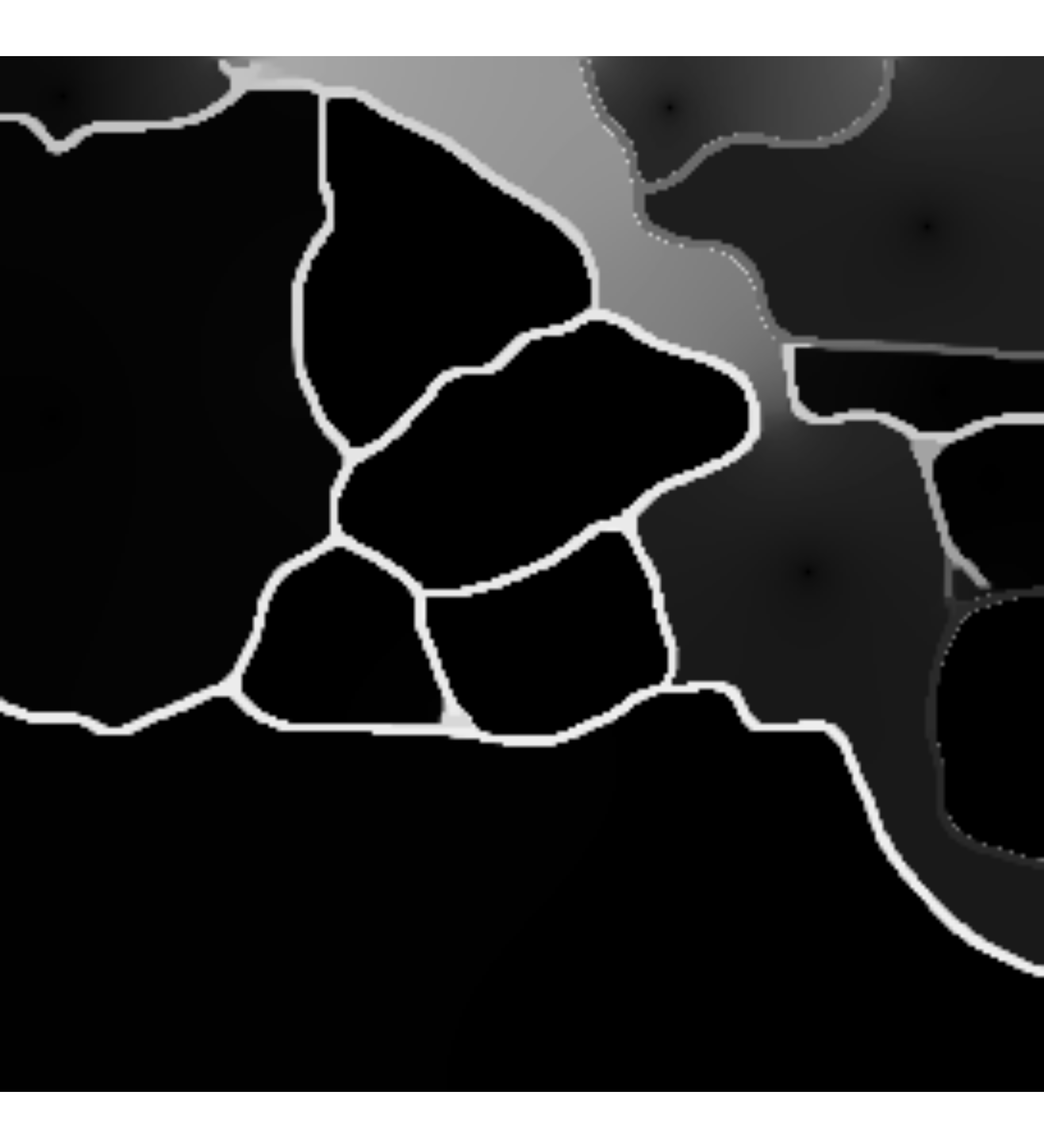}
		\caption{Uncertainty}
		\label{sfig4:example_performance}
\end{subfigure}}
\caption{The Probabilistic Watershed profits from using all spanning forests instead of only the minimum cost one. (\ref{sfig1:example_performance}) Crop of a CREMI image \cite{CREMI} with marked seeds. (\ref{sfig2:example_performance}) and (\ref{sfig3:example_performance}) show results of Watershed and multiple seed  Probabilistic Watershed (end of section \ref{sec:PWS}) applied to edge-weights from \cite{Cerrone_RW}. (\ref{sfig4:example_performance}) shows the entropy of the label probabilities of the Probabilistic Watershed (white high, black low). The Watershed errs in an area where the Probabilistic Watershed expresses uncertainty but is correct.
}
	\label{fig:example_performance}
\end{figure}

\subsection{Related work}

Watershed as a segmentation algorithm was first introduced in \cite{beucher1979watersheds}. Since then it has been studied from different points of view
	\cite{1993Beucher,Topological_Watershed}, notably as a minimum spanning forest that separates the seeds \cite{Cousty2009}. The Random Walker \cite{Grady2006,zhou2004learning,Zhu2003,belkin2006manifold} calculates the probability that a random walker starting at a query node reaches a certain seed before the other ones. Both algorithms are related in \cite{Couprie2011} by a limit consideration termed Power Watershed algorithm. In this work, we establish a different link between the Watershed and the Random Walker. The Watershed's and Random Walker's recent combination with deep learning \cite{Learned_Watershed, vernaza2017learning, Cerrone_RW} also connects our Probabilistic Watershed to deep learning.

Related to our work by name though not in substance is the "Stochastic Watershed" \cite{stochasticwatershed,MALMBERG2014}, which samples different instances of seeds and calculates a probability distribution over segmentation boundaries. Instead, in \cite{Watershed_uncertainty_stimators} the authors suggest sampling the edge-costs in order to define an uncertainty measure of the labeling. They show that it is NP-hard to calculate the probability that a node is assigned to a seed if the edge-costs are stochastic. We derive a closed-form formula for this probability for non-stochastic costs by sampling spanning forests. Ensemble Watersheds proposed by \cite{najman2019}  samples part of the seeds and part of the features which determine the edge-costs. Introducing stochasticity to distance transforms makes a subsequent Watershed segmentation more robust to noise \cite{ofverstedt2019stochastic}. Minimum spanning trees are also applied in optimum-path forest learning, where confidence measures can be computed \cite{Fernandes2019,Fernandes2019Prob_OPF}. Similar to our forest distribution, \cite{DevelopmentsRSP} considers a Gibbs distribution over shortest paths. This approach is extended to more general bags-of-paths in \cite{franccoisse2017bag}.

Entropic regularization has been used most successfully in optimal transport \cite{cuturi2013sinkhorn} to smooth the combinatorial optimization problem and hence afford end-to-end learning in conjunction with deep networks \cite{mensch2018differentiable}. Similarly, we smooth the combinatorial minimum spanning forest problem by considering a Gibbs distribution over all spanning forests.

The matrix tree theorem (MTT) plays a crucial role in our theory, permitting us to measure the weight of a set of forests. The MTT is applied in machine learning \cite{koo2007structured}, biology \cite{Teixeira2015} and network analysis  \cite{Teixeira2013,Tsen1994}. The matrix forest theorem (MFT), a generalization of the MTT, is applied in \cite{Forest_measure,Senelle2014}. By means of the MFT, a distance on the graph is defined in \cite{Forest_measure}. In a similar manner as we do with the MTT, \cite{Senelle2014} is able to compute a Gibbs distribution of forests using the MFT.  

Some of the theoretical results of our work are mentioned in \cite{Grady2006, gradytextbook}, where they refer to \cite{Biggs1997}. In contrast to \cite{Grady2006}, we emphasize the relation with the Watershed and develop the theory in a simpler and more direct way.

\section{Background}
\label{sec:Background}
\subsection{Notation and terminology}
\label{sec:Notation_Terminology}

Let $G=(V,E,w,c)$ be a graph where $V$ denotes the set of nodes, $E$  the set of edges and $w$ and $c$ are functions that assign a weight $w(e)\in \mathbb{R}_{\geq 0}$ and a cost $c(e)\in\mathbb{R}$ to each edge $e\in E$. All the graphs $G$ considered will be connected and undirected. When we speak of a multigraph, we allow for multiple edges incident to the same two nodes but not for self-loops. We will consider simple graphs unless stated otherwise.

The Laplacian of a graph $L\in \mathbb{R}^{|V|\times|V|}$  is defined as  
\[L_{uv}:=\begin{cases}
-w\big(\{u,v\}\big) &\text{ if } u\neq v\\
\sum_{k\in V} w\big(\{u,k\}\big) &\text{ if } u=v
\end{cases},\]
where we consider $w\big(\{u,v\}\big)=0$ if $\{u,v\}\notin E$. $L^+$ will denote its pseudo-inverse.

We define the weight of a graph as the product of the weights of all its edges, $w(G)=\prod_{e\in E}w(e)$. The weight of a set of graphs, $w(\{G_i\}_{i=0}^n)$ is the sum of the weights of the graphs. In a similar manner, we define the cost of a graph as the sum of the costs of all its edges, $c(G)=\sum_{e\in E}c(e)$.

The set of spanning trees of $G$ will be denoted by $\Setspt$. Given a tree $t\in\Setspt$ and nodes $u,v\in V$, the set of edges on the unique path between $u$ and $v$ in $t$ will be denoted by $\mathcal{P}_t(u,v)$. By $\Forest[u][v]$ we denote the set of 2-trees spanning forests, i.e. spanning forests with two trees, such that $u$ and $v$ are not connected. Furthermore, if we consider a third node $q$, we define $\Forest[u][v][q]:=\Forest[u][v]\cap \Forest[q][v]$, i.e. all 2-trees spanning forests such that $q$ and $u$ are in one tree and $v$ belongs to the other tree. Note that the sets $\Forest[u][v][q]\ (=\Forest[q][v][u])$ and $\Forest[v][u][q]\ (=\Forest[q][u][v])$ form a partition of $\Forest[u][v]\ (=\Forest[v][u])$, since $q$ must be connected either to $u$ or $v$, but not to both. In order to shorten the notation we will refer to 2-trees spanning forests simply as 2-forests.

We consider $w(e)=\exp(-\mu c(e))$, $\mu\geq0$, as will be motivated in Section \ref{subsec:Prob_connecting_nodes} by the definition of a Gibbs distribution over the 2-forests in $\Forest[u][v]$. Thus, a low edge-cost corresponds to a large edge-weight, and a minimum edge-cost spanning forest (mSF) is equivalent to a maximum edge-weight spanning forest (MSF).

\subsection{Seeded Watershed as minimum cost spanning forest computation}
\label{subsec:Watershed}
Let $G=(V,E,c)$ be a graph and $c(e)$ be the cost of edge $e$. The lower the cost, the higher the affinity between the nodes incident to $e$. Given different seeds, a forest in the graph defines a segmentation over the nodes as long as each component contains a different seed. The cost of a forest, $c(f)$, is equal to the sum of the costs of its edges. The Watershed algorithm calculates a minimum cost spanning forest, mSF, (or maximum weight, MSF) such that the seeds belong to different components \cite{Cousty2009}.

\subsection{Matrix tree theorem}
In our approach we want to take all possible 2-forests in $\Forest[u][v]$ into account. The probability of a node label will be measured by the cumulative weight of the 2-forests connecting the node to a seed of that label. To compute the weight of a set of 2-forests we will use the matrix tree theorem (MTT) which can be found e.g. in chapter 4 of \cite{Tutte1984} (see Appendix \ref{app:calculus_F})) and has its roots in \cite{kirchhoff1847ueber}.
\begin{theorem}[\textbf{MTT}]
	\thlabel{Th:Matrix_tree}
	For any edge-weighted multigraph $G$ the sum of the weights of the spanning trees of $G$, $w(\Setspt)$, is equal to
	\[w(\Setspt)\coloneqq\sum_{t\in \Setspt }w(t)= \sum_{t\in \Setspt }\prod_{e\in E_t}w(e) =\frac{1}{|V|}\det\Big(L+\frac{1}{|V|}\mathbbm{1}\mathbbm{1}^\top\Big)=\det(L^{[v]}),\]
	where $\mathbbm{1}$ is a column vector of $1$'s. $L^{[v]}$ is the matrix obtained from $L$  after removing the row and column corresponding to an arbitrary but fixed node $v$.
\end{theorem}

This theorem considers trees instead of 2-forests. The key idea to obtain an expression for $w\left(\Forest[u][v]\right)$ by means of the MTT is that any 2-forest $f\in \Forest[u][v]$ can be transformed into a tree by adding an artificial edge $\bar{e}=\{u,v\}$ which connects the two components of $f$ (as done in section 9 of \cite{Biggs1997} or in the original work of Kirchhoff \cite{kirchhoff1847ueber}). We obtain the following lemma, which is proven in Appendix \ref{app:calculus_F}.

\begin{lemma}~	
	\thlabel{lem:T_e_weighted}
	Let $G=(V,E,w)$  be an undirected edge-weighted connected graph and $u,v\in V$  arbitrary vertices.	
	\begin{enumerate}[label=\alph*),leftmargin=*]

    \item Let $\ell_{ij}^+$ denote the entry $ij$ of the pseudo-inverse of the Laplacian of $G$, $L^+$. Then we get
		\begin{equation}
		w(\Forest[u][v])=w(\Setspt)\left(\ell^+_{uu}+\ell^+_{vv}-2\ell^+_{uv}\right).
		\label{eq:T_e_weighted_pseudoinverse}
		\end{equation}
		
		 \item Let $\ell_{ij}^{-1,[r]}$ denote the entry $ij$ of the inverse of the matrix $L^{[r]}$ (the Laplacian $L$ after removing the row and the column corresponding to node $r$), then
		\begin{equation}
		w(\Forest[u][v])=\begin{cases}
		w(\Setspt)\left(\ell_{uu}^{-1,[r]}+\ell_{vv}^{-1,[r]}-2\ell_{uv}^{-1,[r]}\right) &\text{ if } r\neq u,v\\
		w(\Setspt)\ell_{uu}^{-1,[v]} &\text{ if } r=v \text{ and } u\neq v \\
		w(\Setspt)\ell_{vv}^{-1,[u]} &\text{ if } r=u \text{ and } u\neq v .\\
		\end{cases}
		\label{eq:T_e_weighted_remove_r}
		\end{equation}
	\end{enumerate}
\end{lemma}

\subsection{Effective resistance}
In electrical network theory, the circuits are also interpreted as graphs, where the weights of the edges are defined by the reciprocal of the resistances of the circuit. The effective resistance between two nodes $u$ and $v$ can be defined as 
	$\EfR[u][v]:=\left(\nu_{u}-\nu_v\right)/I$
	where $\nu_{u}$ is the potential at node $u$ and $I$ is the current flowing into the network. Other equivalent expressions for the effective resistance \cite{Ghosh2008} in terms of the matrices $L^+$ and $L^{[r]}$, as defined in \thref{lem:T_e_weighted}, are  
	\begin{equation}
	\label{eq:Effective_resistance} 
	\EfR[u][v]=\ell^+_{uu}+\ell^+_{vv}-2\ell^+_{uv}=\begin{cases}
		\left(\ell_{uu}^{-1,[r]}+\ell_{vv}^{-1,[r]}-2\ell_{uv}^{-1,[r]}\right) &\text{ if } r\neq u,v\\
		\ell_{uu}^{-1,[v]} &\text{ if } r=v \text{ and } u\neq v \\
		\ell_{vv}^{-1,[u]} &\text{ if } r=u \text{ and } u\neq v .\\
		\end{cases}
	\end{equation}
	We observe that the expressions in \thref{lem:T_e_weighted} and in equation \eqref{eq:Effective_resistance} are proportional. We will develop this relation further in Section \ref{subsec:Compute_Prob_PWS}. An important property of the effective resistance is that it defines a metric over the nodes of a graph (\cite{Fouss2016} Section 2.5.2).

\section{Probabilistic Watershed}
\label{sec:PWS}

Instead of computing the mSF, as in the Watershed algorithm, we take into account  all the 2-forests that separate two seeds $s_1$ and $s_2$ in two trees according to their costs. Since each 2-forest assigns a query node to exactly one of the two seeds, we calculate the probability of sampling a 2-forest that connects the seed with the query node. Moreover, this provides an uncertainty measure of the assigned label. We call this approach to semi-supervised learning "Probabilistic Watershed".
\subsection{Probability of connecting two nodes in an ensemble of 2-forests}
\label{subsec:Prob_connecting_nodes}

In Section \ref{sec:Notation_Terminology}, we defined the cost of a forest as the cumulative cost of its edges. We assume that the 2-forests $f\in \Forest[s_1][s_2]$ follow a probability distribution that minimizes the expected cost of a 2-forest among all distributions of given entropy $J$. Formally, the 2-forests are sampled from the distribution which minimizes

\begin{equation}
\label{eq:min_entropy}
    \min_{P} \sum_{f\in\Forest[s_1][s_2]} P(f)c(f), \quad \text{s.t.} \quad 
\sum_{f\in\Forest[s_1][s_2]} P(f)=1 \, \text{ and } \,
\mathcal{H}(P)=J,
\end{equation}

where $\mathcal{H}(P)$ is the entropy of $P$. The lower the entropy, the more probability mass is given to the 2-forests of lowest cost. The minimizing distribution is the Gibbs distribution (e.g. \cite{winkler2012image} 3.2):
\begin{equation}
\label{eq:Gibbs_dist}
    P(f)=\frac{\exp\left(-\mu c(f)\right)}{\sum_{f'\in\Forest[s_1][s_2]}\exp\left(-\mu c(f')\right)}=\frac{\prod_{e\in E_{f}}\exp(-\mu c(e))}{\sum_{f'\in\Forest[s_1][s_2]} \prod_{e\in E_{f'}}\exp(-\mu c(e))}=\frac{w(f)}{\sum_{f'\in\Forest[s_1][s_2]} w(f')},
\end{equation}

where $\mu$ implicitly determines the entropy. A higher $\mu$ implies a lower entropy (see Section \ref{sec:Pow_Watershed} and \figurename{} \ref{fig:mu_behaviour} in the appendix). According to \eqref{eq:Gibbs_dist}, an appropriate choice for the edge-weights is $w(e)=\exp(-\mu c(e))$.
The main definition of the paper is:

\begin{definition}[\textbf{Probabilities of the Probabilistic Watershed}]
\thlabel{Def:Prob}
Given two seeds $s_1$ and $s_2$ and a query node $q$, we define the Probabilistic Watershed's probability that $q$ and $s_1$ have the same label as the probability of sampling a 2-forest that connects $s_1$ and $q$, while separating the seeds:
\begin{equation}
\label{def:ProbWS_Prob}
P(q\sim s_1)\coloneqq \sum_{f \in \Forest[s_1, q][s_2]} P(f) = \sum_{f\in \Forest[s_1, q][s_2]} w(f)  \Big/  \sum_{f' \in \Forest[s_1][s_2]} w(f') = w\left(\Forest[s_1, q][s_2]\right) \big/w\left(\Forest[s_1][s_2]\right).  
\end{equation}

\end{definition}

    The Watershed algorithm computes a minimum cost 2-forest, which is the most likely 2-forest according to  \eqref{eq:Gibbs_dist}, and segments the nodes by their connection to seeds in the minimum cost spanning 2-forest. However, it does not indicate which label assignments were ambiguous, for instance due to the existence of other low - but not minimum - cost 2-forests. This makes it a brittle "winner-takes-all" approach. In contrast, the Probabilistic Watershed takes all spanning 2-forests into account according to their cost (see \figurename{} \ref{fig:summary_ProbWS}). The resulting assignment probability of each node provides an uncertainty measure. Assigning each node to the seed for which it has the highest probability can yield a segmentation different from the Watershed's. 

\subsection{Computing the probability of a query being connected to a seed}
\label{subsec:Compute_Prob_PWS}
In the previous subsection, we defined the probability of a node being assigned to a seed via a Gibbs distribution over all exponentially many 2-forests. Here, we show that it can be computed analytically using only elementary graph constructions and the MTT (\thref{Th:Matrix_tree}).
In \thref{lem:T_e_weighted} we have stated how to calculate $w(\Forest[u][v])$ for any $u,v\in V$. Applying this to $\Forest[s_1][s_2]$, $\Forest[s_1][q]$ and $\Forest[s_2][q]$ we can compute $w(\Forest[s_1][s_2][q])$ and $w(\Forest[s_2][s_1][q])$ by means of a linear system.

$\Forest[u][v][q]$ and $\Forest[v][u][q]$ form a partition of $\Forest[u][v]$ for any mutually distinct nodes $u,v,q$ as mentioned in Section \ref{sec:Notation_Terminology}. Thus, we obtain the linear system of three
equations in three unknowns:
\begin{equation}
\label{Lin_syst_bin}
\begin{array}{ccccc}
w(\Forest[s_1][s_2][q])+w(\Forest[s_1][q][s_2])&=&w(\Forest[s_2][q])\\
w(\Forest[s_1][q][s_2])+w(\Forest[s_2][s_1][q])&=&w(\Forest[s_1][q])\\
w(\Forest[s_1][s_2][q])+w(\Forest[s_2][s_1][q])&=&w(\Forest[s_1][s_2]).
\end{array}
\end{equation}

In this paragraph, we describe an alternative way of deriving \eqref{Lin_syst_bin} by relating spanning 2-forests to spanning trees before we solve it in \eqref{eq:sol_system}. This is similar to our use of the MTT for counting spanning 2-forests instead of trees in \thref{lem_app_w(Te)+w(T)=w(Ge)} (see Appendix \ref{app:calculus_F})
Let $t$ be a spanning tree of $G$. To create a 2-forest $f\in\Forest[s_1][s_2]$ from $t$ we need to remove an edge $e$ in the path from $s_1$ to $s_2$, that is $e\in \mathcal{P}_t(s_1,s_2)$. This edge $e$ must be either in $\mathcal{P}_t(q,s_1)\cap\mathcal{P}_t(s_1,s_2)$ or $\mathcal{P}_t(q,s_2)\cap\mathcal{P}_t(s_1,s_2)$ (shown in red and blue respectively in \figurename{} \ref{sfig4:Binary_case}), as the union of $P_t(s_1, q)$ and $P_t(q, s_2)$ contains $P_t(s_1, s_2)$ and removing $e$ from $t$ cannot pairwise separate $q$, $s_1$ and $s_2$. If we remove an edge from $\mathcal{P}_t(q,s_2)\cap\mathcal{P}_t(s_1,s_2)$, we get $f\in \Forest[s_1][s_2][q]$ since we are disconnecting $s_2$ from $q$, otherwise $f \in \Forest[s_2][s_1][q]$. Analogously, we obtain a 2-forest in $\Forest[s_1][q]$ or $\Forest[s_2][q]$ if we remove an edge $e$ from $\mathcal{P}_t(s_1,q)$ or $\mathcal{P}_t(s_2,q)$ respectively (see \figurename{} \ref{fig:Binary_case}). When applied to all spanning trees, we obtain the system \eqref{Lin_syst_bin}.

\tikzset{
	invisible/.style={white,circle,draw,minimum size=0.06cm,inner sep=0pt]},
	main node/.style={circle,draw,minimum size=0.5cm,inner sep=0pt]},
}
\begin{figure}
	\noindent\makebox[\linewidth][c]{	\begin{subfigure}{.25\textwidth}
			\centering
			\begin{tikzpicture}[node distance=1.5cm,
			thick,main node/.style={circle,draw,minimum size=0.3cm,inner sep=0pt]}]
			\def\scale{0.7}
			
			\node[main node] (0) at (\scale*0,\scale*0) {\tiny$q$}; 
			\node[invisible] (4) at (\scale*0,-\scale*1.5)   {}; 
			\node[main node] (2) at (\scale*1,-\scale*2.5) {\tiny$s_2$};
			\node[main node] (1) at (-\scale*1,-\scale*2.5) {\tiny$s_1$};

			\path[dashed,red]
			(4) edge node {} (1);
			\path[dashed,blue]
			(4) edge node {} (2);
			\path[dashed,green!50!black!50]
			(4) edge node {} (0);

			\end{tikzpicture}
			\caption{spanning tree $t\in\Setspt$}
			\label{sfig1:Binary_case}
		\end{subfigure}%
		\begin{subfigure}{.25\textwidth}
			\centering
			\begin{tikzpicture}[node distance=1.5cm,
			thick,main node/.style={circle,draw,minimum size=0.3cm,inner sep=0pt]}]
			\def\scale{0.7}
		
			\node[main node] (0) at (\scale*0,\scale*0) {\tiny$q$}; 
			\node[invisible] (4) at (\scale*0,-\scale*1.5)   {}; 
			\node[main node] (2) at (\scale*1,-\scale*2.5) {\tiny$s_2$};
			\node[main node] (1) at (-\scale*1,-\scale*2.5) {\tiny$s_1$};

			\path[dashed,red]
			(4) edge node {} (1);
			\path[dashed,blue]
			(4) edge node {} (2);
			\path[dashed,green!50!black!50]
			(4) edge node {} (0);
			
			\draw[blue!75!white!100] (\scale*0.26,-\scale*2.16) --  (\scale*0.53,-\scale*1.63) node [above right,pos=0.2][rectangle split,rectangle split parts=2] {\tiny$f\in \Forest[s_1][s_2][q]$}; 
			
			\draw[green!50!black!100] (\scale*0.3,-\scale*1.1) --  (-\scale*0.3,-\scale*1.1)node [ left,pos=0.8][rectangle split,rectangle split parts=2] {\tiny$f\in \Forest[s_1][q][s_2]$};
			\end{tikzpicture}
			\caption{forest $f\in\Forest[s_2][q]$}
			\label{sfig2:Binary_case}
		\end{subfigure}
		\begin{subfigure}{.25\textwidth}
			\centering
			\begin{tikzpicture}[node distance=1.5cm,
			thick,main node/.style={circle,draw,minimum size=0.3cm,inner sep=0pt]}]
			\def\scale{0.7}
			
			\node[main node] (0) at (\scale*0,\scale*0) {\tiny$q$}; 
			\node[invisible] (4) at (\scale*0,-\scale*1.5)   {}; 
			\node[main node] (2) at (\scale*1,-\scale*2.5) {\tiny$s_2$};
			\node[main node] (1) at (-\scale*1,-\scale*2.5) {\tiny$s_1$};
			\path[dashed,red]
			(4) edge node {} (1);
			\path[dashed,blue]
			(4) edge node {} (2);
			\path[dashed,green!50!black!50]
			(4) edge node {} (0);
			
			\draw[red!75!white!100] (-\scale*0.26,-\scale*2.16) --  (-\scale*0.53,-\scale*1.63) node [above left,pos=0.2][rectangle split,rectangle split parts=2] {\tiny$f\in \Forest[s_2][s_1][q]$};
			
			\draw[green!50!black!100] (\scale*0.3,-\scale*1.1) --  (-\scale*0.3,-\scale*1.1) node [right,pos=0.1][rectangle split,rectangle split parts=2] {\tiny$f\in \Forest[s_1][q][s_2]$};
			
			\end{tikzpicture}
			\caption{forest $f\in\Forest[s_1][q]$}
			\label{sfig3:Binary_case}
		\end{subfigure}
		\begin{subfigure}{.25\textwidth}
			\centering
			\begin{tikzpicture}[node distance=1.5cm,
			thick,main node/.style={circle,draw,minimum size=0.3cm,inner sep=0pt]}]
			\def\scale{0.7}
			
			\node[main node] (0) at (\scale*0,\scale*0) {\tiny$q$}; 
			\node[invisible] (4) at (\scale*0,-\scale*1.5)   {}; 
			\node[main node] (2) at (\scale*1,-\scale*2.5) {\tiny$s_2$};
			\node[main node] (1) at (-\scale*1,-\scale*2.5) {\tiny$s_1$};

			\path[dashed,red]
			(4) edge node {} (1);
			\path[dashed,blue]
			(4) edge node {} (2);
			\path[dashed,green!50!black!50]
			(4) edge node {} (0);
			\draw[red!75!white!100] (-\scale*0.26,-\scale*2.16) --  (-\scale*0.53,-\scale*1.63) node [above left,pos=0.2][rectangle split,rectangle split parts=2] {\tiny$f\in \Forest[s_2][s_1][q]$}; 
			
			\draw[blue!75!white!100] (\scale*0.26,-\scale*2.16) --  (\scale*0.53,-\scale*1.63) node [above right,pos=0.2][rectangle split,rectangle split parts=2] {\tiny$f\in \Forest[s_1][s_2][q]$}; ;
			
			\end{tikzpicture}
			\caption{forest $f\in\Forest[s_1][s_2]$}
			\label{sfig4:Binary_case}
		\end{subfigure}
	}
	\caption{Amongst all spanning forests that isolate seed $s_1$ from $s_2$, we want to identify the fraction of forests connecting $s_1$ and $q$ (\thref{Def:Prob}). The dashed lines represent all spanning trees. Either cut in (\ref{sfig2:Binary_case}) yields a forest separating $q$ from $s_2$. The blue ones are of interest to us. Diagrams (\ref{sfig2:Binary_case}) - (\ref{sfig4:Binary_case}) correspond to the three equations in the linear system \eqref{Lin_syst_bin}, which can be solved for $w(\Forest[s_1][s_2][q])$.}
	\label{fig:Binary_case}
\end{figure}
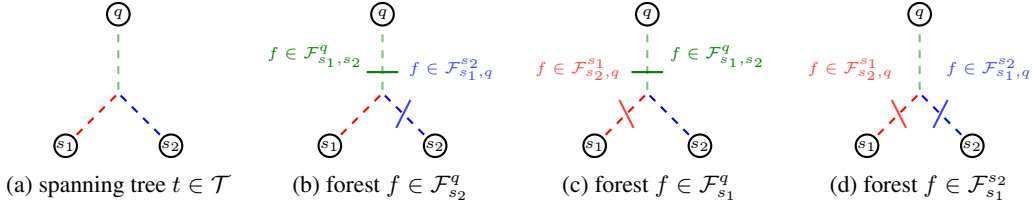

Solving the linear system \eqref{Lin_syst_bin} we obtain \footnote{Section IV.B of \cite{Grady2006} states $w(\Forest[u][v][q])=w(\Forest[u][v])-w(\Forest[v][q])$ for any $u,v,q \in V$ but that formula is incorrect. For instance, it does not hold for the complete graph with nodes $\{u,v,q\}$ and with $w(e)=1$ for all edges $e$, since $w(\Forest[u][v][q])=1 \neq 0= 2-2=w(\Forest[u][v])-w(\Forest[v][q])$.}
\begin{equation}
\label{eq:sol_system}
\begin{aligned}
w\left(\Forest[s_1][s_2][q]\right)=\big(w(\Forest[s_2][q])+w(\Forest[s_1][s_2])-w(\Forest[s_1][q])\big)\big/2.
\end{aligned}
\end{equation}
 In consequence of equation \eqref{eq:sol_system} and \thref{Def:Prob}
 we get the following theorem:
\begin{theorem}
	\thlabel{th:Prob}
	The probability that $q$ has the same label as seed $s_1$ is
	\[P(q \sim s_1)=\big(w(\Forest[s_2][q])+w(\Forest[s_1][s_2])-w(\Forest[s_1][q])\big)\big/\big(2w(\Forest[s_1][s_2])\big).\]
\end{theorem}

\thref{th:Prob} expresses $P(q \sim s_1)$ in terms of weights of 2-forests, which we can compute with \thref{lem:T_e_weighted}, which is based on the MTT. We use this expression to relate $P(q \sim s_1)$ to the effective resistance. As a result of \thref{lem:T_e_weighted} and equation \eqref{eq:Effective_resistance}, for any nodes $u,v\in V$ we have 
\begin{equation}
\label{eq:EfR=w(Te)/w(E)}
\EfR[u][v]=w\left(\Forest[u][v]\right)/w(\Setspt).
\end{equation}
 This relation has already been proven in \cite{Biggs1997} (Proposition 17.1) but in terms of the effective conductance (the inverse of the  effective resistance). Due to $\EfR[u][v]$ being a metric, $w\left(\Forest[u][v]\right)$ also defines  a metric over the nodes of the graph. Combining \eqref{eq:EfR=w(Te)/w(E)} with \thref{th:Prob}, we have that the probability of $q$ having seed $s_1$'s label is
\begin{equation}
\label{eq:triangle_inequality_diff_prob}
P(q \sim s_1)=\big(\EfR[s_2][q]+\EfR[s_2][s_1]-\EfR[s_1][q]\big)\big/\big(2\EfR[s_1][s_2]\big)
\end{equation}
The probability is proportional to the gap in the triangle inequality $\EfR[s_1][q]\leq \EfR[s_1][s_2]+\EfR[s_2][q].$
It will be shown in Section \ref{sec:Relation_Random_Walk} that the probability defined in \thref{Def:Prob} is equal to the probability given by the Random Walker \cite{Grady2006}. Equation \eqref{eq:triangle_inequality_diff_prob} gives an interpretation of this probability, which is new to the best of our knowledge.  We can see that the greater the gap in the triangle inequality, the greater is the probability. Further, we get $
    P(q \sim s_1)\geq P(q \sim s_2)\iff \EfR[s_1][q]\leq \EfR[s_2][q].$
This relation has already been pointed out in \cite{Grady2006} (section IV.B) in terms of the effective conductance between two nodes, but not as explicitly as in \eqref{eq:triangle_inequality_diff_prob}. We note that any metric distance on the nodes of a graph, e.g. the ones mentioned in the introduction, can define an assignment probability along the lines of equation \eqref{eq:triangle_inequality_diff_prob}.

Our discussion was constrained to the case of two seeds only to ease our explanation. We can reduce the case of multiple seeds per label to the two seed case by merging all nodes seeded with the same label. Similarly, the case of more than two labels can be reduced to the two label scenario by using a one versus all strategy: We choose one label and merge the seeds of other labels into one unique seed. In both cases we might introduce multiple edges between node pairs. While having formulated our arguments for simple graphs, they are also valid for multigraphs (see Appendix \ref{app:calculus_F}).

\section{Connection between the Probabilistic Watershed and the 
	Random Walker}
\label{sec:Relation_Random_Walk}

In this section we will show that the Random Walker of \cite{Grady2006} is equivalent to our Probabilistic Watershed, both computationally and in terms of the resulting label probabilities.

\begin{theorem}
	\thlabel{theorem:Trees_equiv_Random_walk_2_seeds} 
	The probability $x_{q}^{s_1}$ that a random walker as defined in \cite{Grady2006} starting at node $q$ reaches $s_1$ first before reaching $s_2$ is equal to the Probabilistic Watershed probability defined in \thref{Def:Prob}:
	\[x_{q}^{s_1}=P(q\sim s_1).\]
\end{theorem}

This equivalence, which we prove in Appendix \ref{App:Proof_ProbWS=RW}, was pointed out by Leo Grady in \cite{Grady2006} section IV.B but with a different approach. Grady relied on results from \cite{Biggs1997}, where potential theory is used. There it is shown that \mbox{ $x_{q}^{s_1}=w(\Forest[s_1][s_2][q])/\left(\EfR[s_1][s_2] w(\Setspt)\right)$.} From this formula we get \thref{theorem:Trees_equiv_Random_walk_2_seeds} by using equation \eqref{eq:EfR=w(Te)/w(E)}:
\[x_{q}^{s_1}=w(\Forest[s_1][s_2][q])/\left(\EfR[s_1][s_2] w(\Setspt)\right)=w\left(\Forest[s_1][s_2][q]\right)/w\left(\Forest[s_1][s_2]\right)=P(q\sim s_1).\]
We have proven the same statement with elementary arguments and without the main theory of \cite{Biggs1997}. Through the use of the MTT, we have shown that the forest-sampling point of view is computationally equivalent to the in practice very useful Random Walker (see \cite{Zhu2003, Grady2006}, and recently \cite{vernaza2017learning, Bui2018, Cerrone_RW, Liu2018, bockelmann2019}), making our method just as potent. We thus refrained from adding further experiments and instead include a new interpretation of the Power Watershed within our framework.

\section{Power Watershed counts minimum cost spanning forests}
\label{sec:Pow_Watershed}
\begin{figure}[h]
\captionsetup[subfigure]{justification=centering}
    \begin{subfigure}{0.49\textwidth}
    \centering
    \includegraphics[height=5cm]{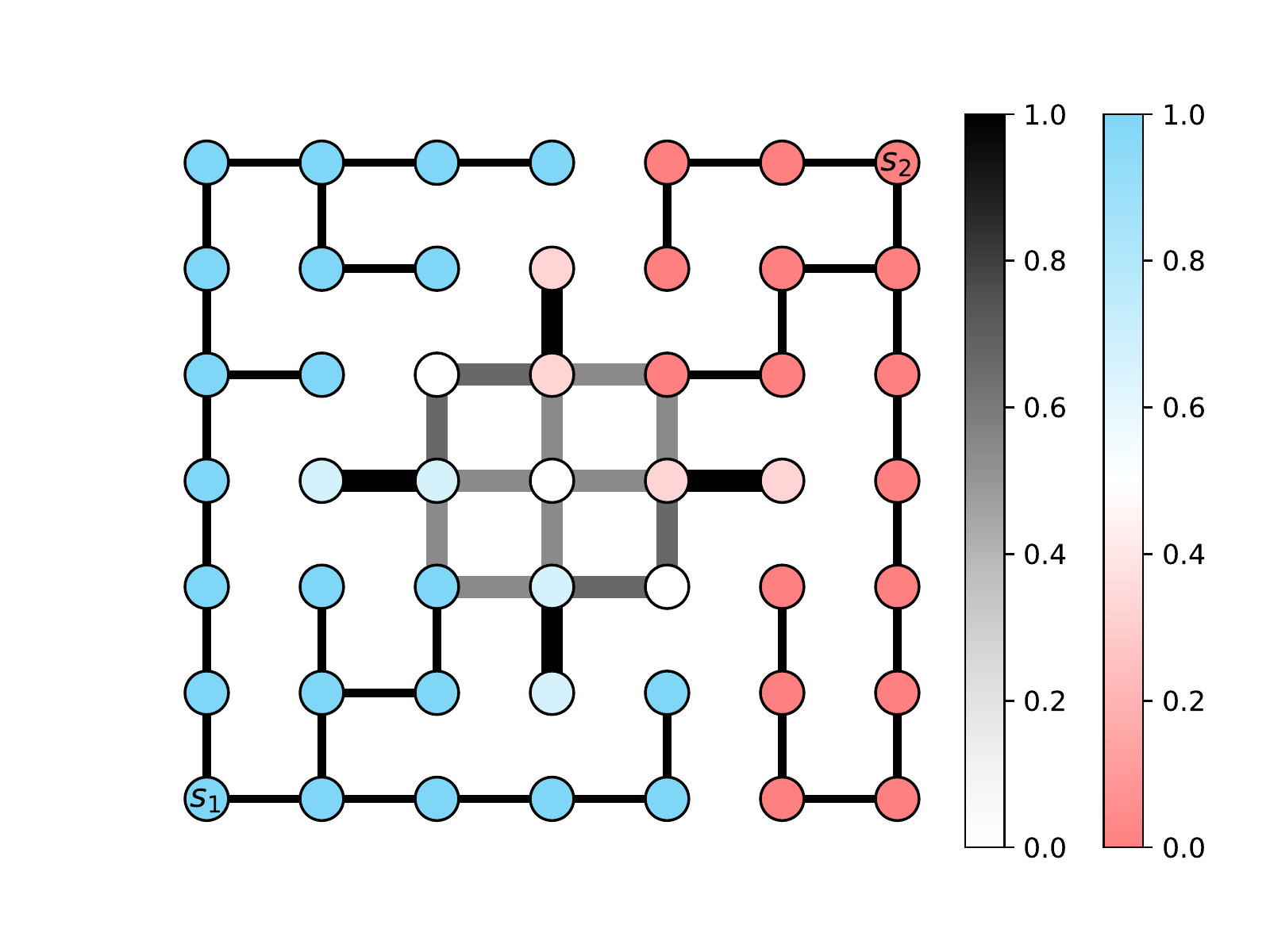}
    \caption{$P(\text{node} \sim s_1)$ and \\ $P(\text{edge} \in \text{some mSF}$)}
    \label{sfig:edge_present}
    \end{subfigure}
    \begin{subfigure}{0.49\textwidth}
    \centering
    \includegraphics[height=5cm]{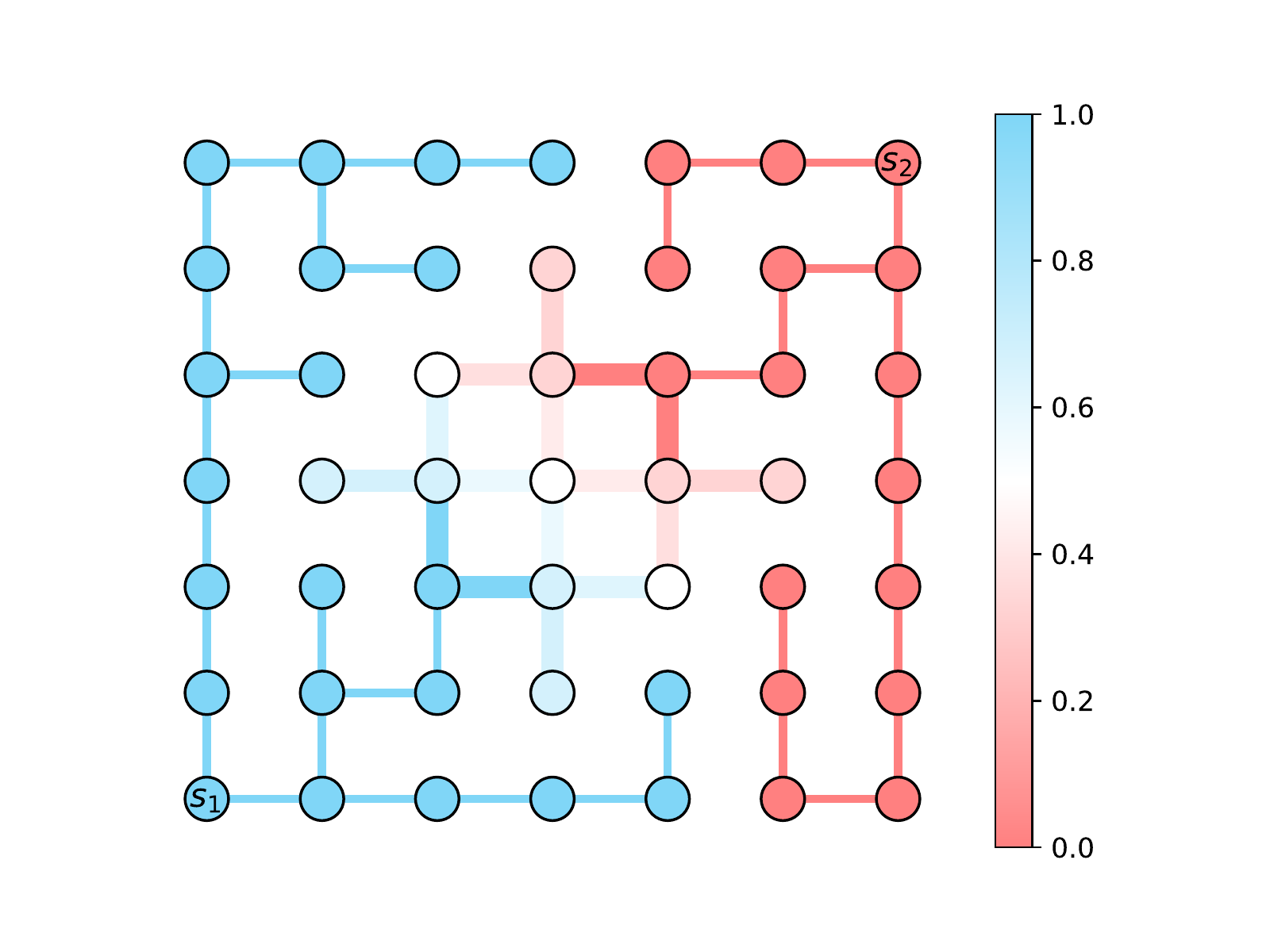}
    \caption{$P(\text{node} \sim s_1)$ and \\ $P(\text{edge} \sim s_1 | \text{edge} \in \text{some msF})$}
    \label{sfig:edge_given_present}
    \end{subfigure}
    \begin{subfigure}{0.49\textwidth}
    \centering
    \includegraphics[height=5cm]{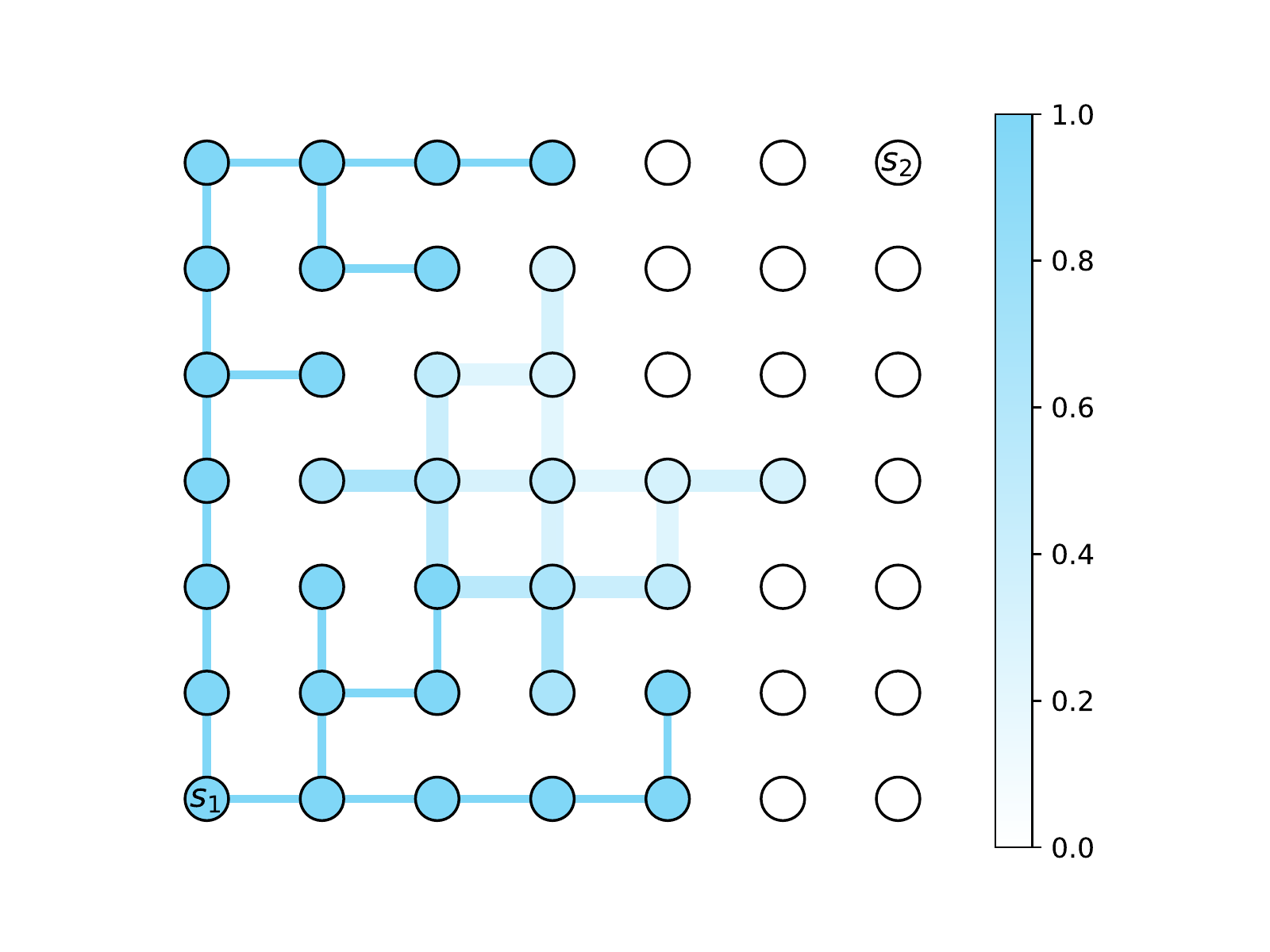}
    \caption{$P(\text{node} \sim s_1)$ and\\
    $P(\text{edge} \sim s_1, \text{edge} \in \text{some mSF})$}
    \label{sfig:edge_s_1}
    \end{subfigure}
    \begin{subfigure}{0.49\textwidth}
    \centering
    \includegraphics[height=5cm]{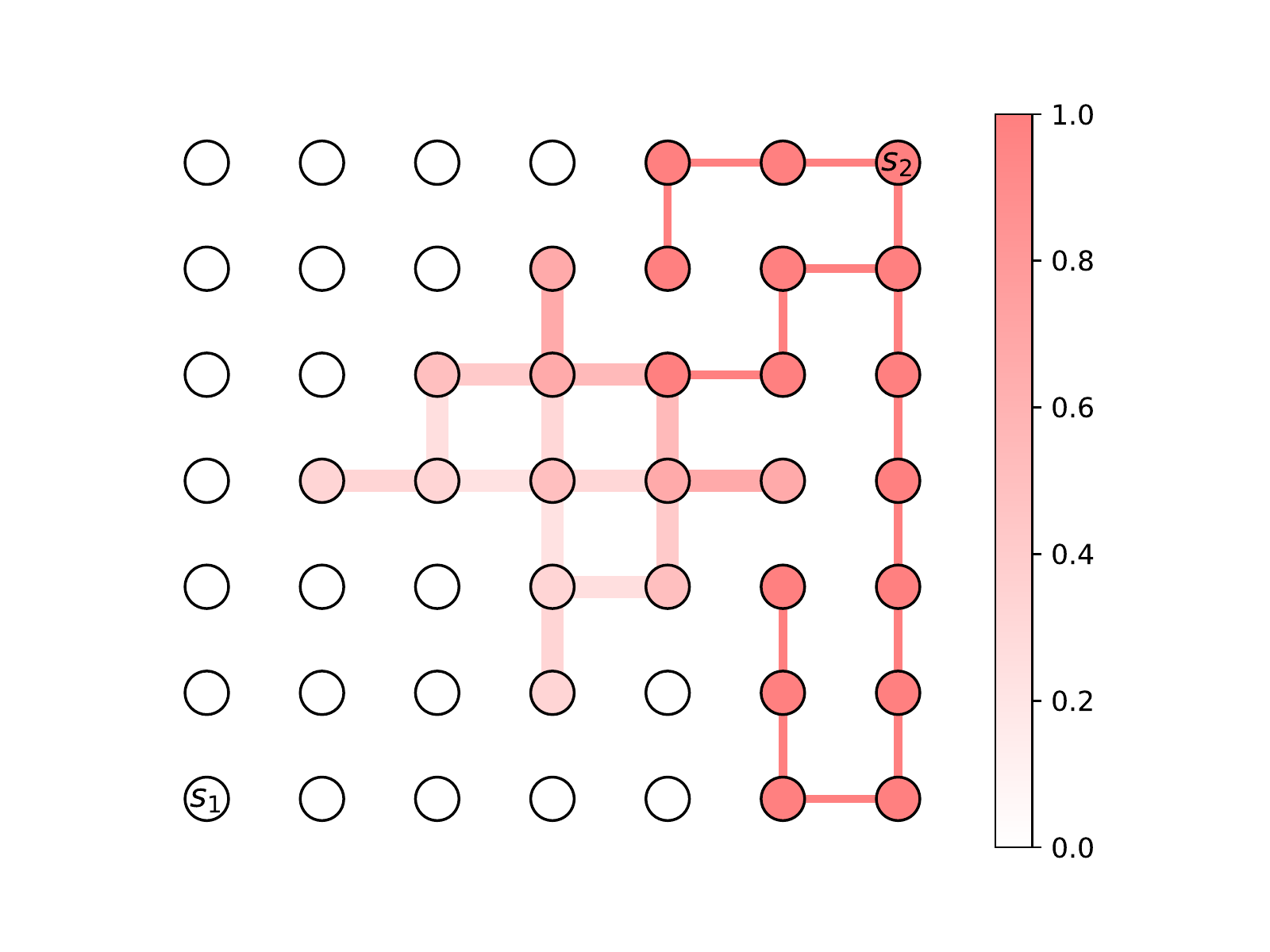}
    \caption{$P(\text{node} \sim s_2)$ and\\
    $P(\text{edge} \sim s_2, \text{edge} \in \text{some mSF})$}
    \label{sfig:edge_s_2}
    \end{subfigure}
    \caption{Power Watershed result on a grid graph with seeds $s_1$, $s_2$ and with random edge-costs outside a plateau of edges with the same cost (wide edges). By the results in Theorem 5.1, the Power Watershed counts mSFs. This is illustrated both with the node- and edge-colors. (\ref{sfig:edge_present}-\ref{sfig:edge_s_2}) The nodes are colored by their probability of belonging to seed $s_1$ ($s_2$), i.e. by the share of mSFs that connect a given node to $s_1$ ($s_2)$. (\ref{sfig:edge_present}) The edge-color indicates the share of mSFs in which the edge is present. (\ref{sfig:edge_given_present}) The edge-color indicates the share of mSFs in which the edge is connected to seed $s_1$ among the mSFs that contain the edge. (\ref{sfig:edge_s_1} - \ref{sfig:edge_s_2}) The edge-color indicates the share of mSFs in which the edge is connected to $s_1$ or $s_2$, respectively, among all mSFs. See Appendix \ref{App:Edge_prob} for a more detailed explanation.
    }
    \label{fig:edge_presence}
\end{figure}
The objective of this section is to recall the Power Watershed \cite{Couprie2011}
(see Appendix \ref{app_power_watershed}
for a summary)
and develop a new understanding of its nature. Power Watershed is a limit over the Random Walker and thus over the equivalent Probabilistic Watershed. The latter's idea of measuring the weight of a set of 2-forests carries over nicely to the Power Watershed, where, as a limit, only the maximum weight / minimum cost spanning forests are considered. This section details the connection.

Let  $G=(V,E,w,c)$ and $s_1,s_2\in V$ be as before. In \cite{Couprie2011} the following objective function is proposed:
\begin{equation}
\label{eq:alg_Couprie}
\arg\min_x \sum_{e=\{u,v\}\in E}\left(w(e)\right)^{\alpha}\left(|x_u-x_v|\right)^{\beta},\ \text{s.t.} \  x_{s_1}=1, \ x_{s_2}=0.
\end{equation}
For $\alpha=1$ and $\beta=2$ it gives the Random Walker's objective function. The Power Watershed considers the limit case when $\alpha\to\infty$ and $\beta$ remains finite.

In section \ref{subsec:Prob_connecting_nodes} we defined the weight of an edge $e$ as $w(e)=\exp(-\mu c(e))$, where $c(e)$ was the edge-cost and $\mu$ implicitly determined the entropy of the 2-forest distribution. By raising the weight of the edges to $\alpha$ we obtain  $\big(w(e)\big)^{\alpha}=\exp(-\mu\alpha c(e))=\exp(-\mu_{\alpha} c(e))$, where $\mu_{\alpha}\coloneqq\mu\alpha$. Therefore, we can absorb $\alpha$ into $\mu$. When $\alpha\to\infty$ (and therefore $\mu_{\alpha}\to\infty$) the distribution will have a lowest entropy. As a consequence only the mSFs / MSFs are considered in the Power Watershed:

\begin{theorem}
	\thlabel{th:Pow_Wat=Max_Span_tree_ratio}
	Given two seeds $s_1$ and $s_2$, let us denote the potential of node $q$ being assigned to seed $s_1$ by the Power Watershed with $\beta=2$ as $x_{q}^{\text{PW}}$. Let further $ w_{\max}$ be $\max_{f\in \Forest[s_1][s_2]} w(f)$. Then
	\[x_{q}^{\text{PW}}=\frac{\left|\{f\in \Forest[s_1][s_2][q] \ : \  w(f)=w_{\max}\}\right|}{\left|\{f\in \Forest[s_1][s_2] \ : \  w(f)=w_{\max}\}\right|}
	.\]
\end{theorem}

\tikzset{
	dot/.style 2 args={fill, circle, inner sep=0pt, label={#1:\scriptsize #2}},	
	fulldot/.style 2 args={circle,draw,minimum size=0.2cm,inner sep=0pt, label={#1:\scriptsize #2}},
	main node/.style={circle,draw,minimum size=0.3cm,inner sep=0pt]},
	small node/.style={circle,draw,minimum size=0.1cm,inner sep=0pt]},
}

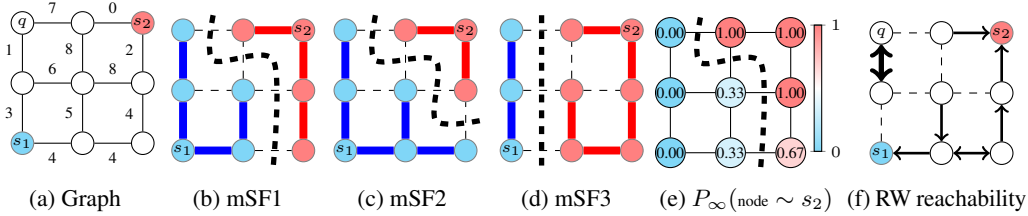
\begin{figure*}
	\centering
	\noindent
		\makebox[0.8\linewidth][c]{\begin{subfigure}[t]{.16\textwidth}
				\centering
				\begin{tikzpicture}[]
				\def\scale{0.8 }
				
				\node[main node,opacity=.5,black,fill=cyan,text opacity=1] (0) at (0,0) {\tiny$s_1$};
				\node[main node] (1) at (\scale*1,0) {};
				\node[main node] (2) at (\scale*2,0) {};
				\node[main node] (3) at (0,\scale*1) {};
				\node[main node] (4) at (\scale*1,\scale*1) {};
				\node[main node] (5) at (\scale*2,\scale*1) {};
				\node[main node] (6) at (0,\scale*2) {\tiny$q$};
				\node[main node] (7) at (\scale*1,\scale*2) {};
				\node[main node,opacity=.5,black,fill=red,text opacity=1] (8) at (\scale*2,\scale*2) {\tiny$s_2$};
				
				\path[-,draw]
				(0) edge node[below] {\tiny4} (1)
				(0) edge node[left] {\tiny3} (3)
				
				(1) edge node[below] {\tiny4} (2)
				(1) edge node[left] {\tiny5} (4)
				
				(2) edge node[left] {\tiny4} (5)

				(3) edge node[above] {\tiny6} (4)
				(3) edge node[above left,pos=0.2] {\tiny1} (6)
				
				(4) edge node[above] {\tiny 8} (5)
				(4) edge node[above left,pos=0.2] {\tiny8} (7)

				(5) edge node[above left,pos=0.2] {\tiny2} (8)
				
				(6) edge node[above] {\tiny7} (7)
				
				(7) edge node[above] {\tiny0} (8);
				\end{tikzpicture}
				\caption{Graph}
				\label{sfig1:Count_MSF}
			\end{subfigure}%
			\begin{subfigure}[t]{.15\textwidth}
				\centering
				\begin{tikzpicture}
                    
				\def\scale{0.8 }

				\node[main node,opacity=.5,black,fill=cyan,text opacity=1] (0) at (0,0) {\tiny$s_1$};
				\node[main node,opacity=.5,black,fill=cyan,text opacity=1] (1) at (\scale*1,0) {};
				\node[main node,opacity=.5,black,fill=red,text opacity=1] (2) at (\scale*2,0) {};
				\node[main node,opacity=.5,black,fill=cyan,text opacity=1] (3) at (0,\scale*1) {};
				\node[main node,opacity=.5,black,fill=cyan,text opacity=1] (4) at (\scale*1,\scale*1) {};
				\node[main node,opacity=.5,black,fill=red,text opacity=1] (5) at (\scale*2,\scale*1) {};
				\node[main node,opacity=.5,black,fill=cyan,text opacity=1] (6) at (0,\scale*2) {};
				\node[main node,opacity=.5,black,fill=red,text opacity=1] (7) at (\scale*1,\scale*2) {};
				\node[main node,opacity=.5,black,fill=red,text opacity=1] (8) at (\scale*2,\scale*2) {\tiny$s_2$};

				\path[-,draw,red,line width=3pt]
				
				(2) edge node[left] {} (5)
				
				(5) edge node[left] {} (8)
				
				(7) edge node[above] {} (8);
				
				\path[dashed,draw]
                (1) edge node[below] {} (2)
                
				(3) edge node[above] {} (4)
				
				(4) edge node[above] {} (5)
				
				(4) edge node[above left] {} (7)
				
				(6) edge node[above] {} (7);

				\path[-,draw,blue,line width=3pt]
				
				(0) edge node[below] {} (1)
				
				(0) edge node[left] {} (3)
				
				(1) edge node[below left] {} (4)

				(3) edge node[left] {} (6);

    			\node[cut node] (c0) at (0+0.5*\scale,\scale*0.25+\scale*2) {};
    			\node[cut node] (c1) at (0+0.5*\scale,\scale*0.5+\scale*1) {};
				\node[cut node] (c2) at (1*\scale+0.5*\scale,\scale*0.5+\scale*1) {};

				\node[cut node] (c3) at (1*\scale+0.5*\scale,-\scale*0.25+\scale*0) {};
				
				\draw [dashed,line width=2pt,xshift=4cm]  plot[smooth] coordinates { (c0) (c1) (c2) (c3) };
                
				\end{tikzpicture}
				\caption{mSF1}
				\label{sfig2:Count_MSF}
			\end{subfigure}
			\begin{subfigure}[t]{.15\textwidth}
				\centering
				\begin{tikzpicture}
                    
				\def\scale{0.8}

				\node[main node,opacity=.5,black,fill=cyan,text opacity=1] (0) at (0,0) {\tiny$s_1$};
				\node[main node,opacity=.5,black,fill=cyan,text opacity=1] (1) at (\scale*1,0) {};
				\node[main node,opacity=.5,black,fill=cyan,text opacity=1] (2) at (\scale*2,0) {};
				\node[main node,opacity=.5,black,fill=cyan,text opacity=1] (3) at (0,\scale*1) {};
				\node[main node,opacity=.5,black,fill=cyan,text opacity=1] (4) at (\scale*1,\scale*1) {};
				\node[main node,opacity=.5,black,fill=red,text opacity=1] (5) at (\scale*2,\scale*1) {};
				\node[main node,opacity=.5,black,fill=cyan,text opacity=1] (6) at (0,\scale*2) {};
				\node[main node,opacity=.5,black,fill=red,text opacity=1] (7) at (\scale*1,\scale*2) {};
				\node[main node,opacity=.5,black,fill=red,text opacity=1] (8) at (\scale*2,\scale*2) {\tiny$s_2$};

				\path[-,draw,red,line width=3pt]
				
				(5) edge node[left] {} (8)
				
				(7) edge node[above] {} (8);
				
				\path[dashed,draw]
				
				(2) edge node[left] {} (5)
                
				(3) edge node[above] {} (4)
				
				(4) edge node[above] {} (5)
				
				(4) edge node[above left] {} (7)
				
				(6) edge node[above] {} (7);

				\path[-,draw,blue,line width=3pt]
				
				(0) edge node[below] {} (1)
				
				(0) edge node[left] {} (3)
				
                (1) edge node[below] {} (2)
				
				(1) edge node[below left] {} (4)

				(3) edge node[left] {} (6);

    			\node[cut node] (c0) at (0+0.5*\scale,\scale*0.25+\scale*2) {};
    			\node[cut node] (c1) at (0+0.5*\scale,\scale*0.5+\scale*1) {};
				\node[cut node] (c2) at (1*\scale+0.5*\scale,\scale*0.5+\scale*1) {};
				
                \node[cut node] (c3) at (1*\scale+0.5*\scale,\scale*0.5+\scale*0) {};
				\node[cut node] (c4) at (2*\scale+0.25*\scale,\scale*0.5+\scale*0) {};
				
				\draw [dashed,line width=2pt,xshift=4cm]  plot[smooth] coordinates { (c0) (c1) (c2) (c3) (c4) };
                
				\end{tikzpicture}
				\caption{mSF2}
				\label{sfig3:Count_MSF}
			\end{subfigure}
			\begin{subfigure}[t]{.15\textwidth}
				\centering
				\begin{tikzpicture}
                    
				\def\scale{0.8 }

				\node[main node,opacity=.5,black,fill=cyan,text opacity=1] (0) at (0,0) {\tiny$s_1$};
				\node[main node,opacity=.5,black,fill=red,text opacity=1] (1) at (\scale*1,0) {};
				\node[main node,opacity=.5,black,fill=red,text opacity=1] (2) at (\scale*2,0) {};
				\node[main node,opacity=.5,black,fill=cyan,text opacity=1] (3) at (0,\scale*1) {};
				\node[main node,opacity=.5,black,fill=red,text opacity=1] (4) at (\scale*1,\scale*1) {};
				\node[main node,opacity=.5,black,fill=red,text opacity=1] (5) at (\scale*2,\scale*1) {};
				\node[main node,opacity=.5,black,fill=cyan,text opacity=1] (6) at (0,\scale*2) {};
				\node[main node,opacity=.5,black,fill=red,text opacity=1] (7) at (\scale*1,\scale*2) {};
				\node[main node,opacity=.5,black,fill=red,text opacity=1] (8) at (\scale*2,\scale*2) {\tiny$s_2$};

				\path[-,draw,red,line width=3pt]
				
                (1) edge node[below] {} (2)
				
				(1) edge node[below left] {} (4)
				
				(2) edge node[left] {} (5)
				
				(5) edge node[left] {} (8)
				
				(7) edge node[above] {} (8);
				
				\path[dashed,draw]
				
				(0) edge node[below] {} (1)
                
				(3) edge node[above] {} (4)
				
				(4) edge node[above] {} (5)
				
				(4) edge node[above left] {} (7)
				
				(6) edge node[above] {} (7);

				\path[-,draw,blue,line width=3pt]
				
				(0) edge node[left] {} (3)

				(3) edge node[left] {} (6);
				
				\node[cut node] (c0) at (0+0.5*\scale,\scale*0.25+\scale*2) {};
    			\node[cut node] (c1) at (0+0.5*\scale,-\scale*0.25+\scale*0) {};
				
				\draw [dashed,line width=2pt,xshift=4cm]  plot[smooth,tension=1] coordinates { (c0) (c1) };

				\end{tikzpicture}
				\caption{mSF3}
				\label{sfig4:Count_MSF}
			\end{subfigure}
			\begin{subfigure}[t]{.18\textwidth}
				\centering
				\begin{tikzpicture}
                    
				\def\scale{0.8 }
                \node[main node,opacity=1,black,fill={rgb,255:red,127.5; green,214.5; blue,247},text opacity=1] (0) at (0,0) {\tiny0.00};
				\node[main node,opacity=1,black,fill={rgb,255:red,211.65; green,241.23; blue,252.28},text opacity=1] (1) at (\scale*1,0) {\tiny0.33};
				\node[main node,opacity=1,black,fill={rgb,255:red,255; green,211.65; blue,211.65 },text opacity=1] (2) at (\scale*2,0) {\tiny0.67};
				
				\node[main node,opacity=1,black,fill={rgb,255:red,127.5; green,214.5; blue,247},text opacity=1] (3) at (0,\scale*1) {\tiny0.00};
				\node[main node,opacity=1,black,fill={rgb,255:red,211.65; green,241.23; blue,252.28},text opacity=1] (4) at (\scale*1,\scale*1) {\tiny0.33};
				\node[main node,opacity=1,black,fill={rgb,255:red,255; green,127.5; blue,127.5},text opacity=1] (5) at (\scale*2,\scale*1) {\tiny1.00};
				
				\node[main node,opacity=1,black,fill={rgb,255:red,127.5; green,214.5; blue,247},text opacity=1] (6) at (0,\scale*2) {\tiny0.00};
				\node[main node,opacity=1,black,fill={rgb,255:red,255; green,127.5; blue,127.5},text opacity=1] (7) at (\scale*1,\scale*2) {\tiny1.00};
				\node[main node,opacity=1,black,fill={rgb,255:red,255; green,127.5; blue,127.5},text opacity=1] (8) at (\scale*2,\scale*2) {\tiny1.00};
				
				\path[-,draw]
				(0) edge node[below] {} (1)
				(0) edge node[left] {} (3)
				
				(1) edge node[below] {} (2)
				(1) edge node[below left] {} (4)
				
				(2) edge node[left] {} (5)

				(3) edge node[above] {} (4)
				(3) edge node[left] {} (6)
				
				(4) edge node[above] {} (5)
				(4) edge node[above left] {} (7)

				(5) edge node[left] {} (8)
				
				(6) edge node[above] {} (7)
				
				(7) edge node[above] {} (8);
    			\begin{axis}[
    			hide axis,
    			scale only axis,
    			width=20pt,
    			colormap={cyanred}{
    				rgb255(0)=(127.5,214.5,247)
    				rgb255(1)=(255,255,255)
    				rgb255(2)=(255,127.5,127.5) 
    			},
    			colorbar horizontal,
    			point meta min=0,
    			point meta max=1,
    			colorbar style={
    				at={(\scale*3.4,\scale*3.5)},
    				width=\scale*2.1cm,
    				height=0.1cm,
    				rotate=90,
    				xtick={0,1},
    			},ticklabel style={right,font=\tiny},]
    			\addplot [draw=none] coordinates {(1,8) (1,1)};
    			\vspace{0.12cm}
    			\end{axis}

    			\node[cut node] (c0) at (0+0.5*\scale,\scale*0.25+\scale*2) {};
    			\node[cut node] (c1) at (0+0.5*\scale,\scale*0.5+\scale*1) {};
				\node[cut node] (c2) at (1*\scale+0.5*\scale,\scale*0.5+\scale*1) {};

				\node[cut node] (c3) at (1*\scale+0.5*\scale,-\scale*0.25+\scale*0) {};
				
				\draw [dashed,line width=2pt,xshift=4cm]  plot[smooth] coordinates { (c0) (c1) (c2) (c3) };
                
				\end{tikzpicture}
				\caption{$P_{\infty}(\text{\tiny node} \sim s_2)$}
				\label{sfig5:Count_MSF}
			\end{subfigure}
			
			\begin{subfigure}[t]{.17\textwidth}
				\centering
				\begin{tikzpicture}[]
				\def\scale{0.8 }
				
				\node[main node,opacity=.5,black,fill=cyan,text opacity=1] (0) at (0,0) {\tiny$s_1$};
				\node[main node] (1) at (\scale*1,0) {};
				\node[main node] (2) at (\scale*2,0) {};
				\node[main node] (3) at (0,\scale*1) {};
				\node[main node] (4) at (\scale*1,\scale*1) {};
				\node[main node] (5) at (\scale*2,\scale*1) {};
				\node[main node] (6) at (0,\scale*2) {\tiny$q$};
				\node[main node] (7) at (\scale*1,\scale*2) {};
				\node[main node,opacity=.5,black,fill=red,text opacity=1] (8) at (\scale*2,\scale*2) {\tiny$s_2$};

				\path[dashed,draw]
				
				(0) edge node[left] {} (3)
				
				(3) edge node[above] {} (4)
				
				(4) edge node[above] {} (5)
				
				(4) edge node[above left] {} (7)
				
				(6) edge node[above] {} (7);
				
				\path[->,draw,line width=1pt]
				
				(2) edge node[left] {} (5)
				
			    (5) edge node[left] {} (8)
				
				(7) edge node[above] {} (8);
				
				\path[<-,draw,line width=1pt]
				
				(0) edge node[below] {} (1)
				
				(1) edge node[below left] {} (4);

				\path[<->,draw,line width=1pt]
				
				(1) edge node[below] {} (2);
				
				\path[<->,draw,line width=2pt]
				
				(3) edge node[left] {} (6);
				
				\end{tikzpicture}
				\caption{RW reachability}
				\label{sfig6:Count_MSF}
			\end{subfigure}%
			}
\caption{Forest-interpretation of Power Watershed. (\ref{sfig1:Count_MSF}) Graph with edge-costs and its mSFs in ((\ref{sfig2:Count_MSF})-(\ref{sfig4:Count_MSF})). (\ref{sfig5:Count_MSF}) Power Watershed probabilities for assigning a node to $s_2$. The Power Watershed computes the ratio between the mSFs connecting a node to $s_2$ and all possible mSFs. The dashed lines indicate the segmentation's cut. (\ref{sfig6:Count_MSF}) indicates the allowed Random Walker transitions  when $\mu\to\infty$ with headed arrows. The Random Walker interpretation of the Power Watershed breaks down in the limit case since a Random Walker starting at node $q$ does not reach any seed, but oscillates along the bold arrow.}
\label{fig:Count_MSF}
	
\end{figure*}

\thref{th:Pow_Wat=Max_Span_tree_ratio},which we prove in Appendix \ref{App:Proof_PWS_count_mSF}, interprets the Power Watershed potentials as a ratio of 2-forests similar to the Probabilistic Watershed. But instead of all 2-forests the Power Watershed only considers minimum cost 2-forests (equivalently maximum weight 2-forests) as they are the only ones that matter after taking the limit $\mu \to \infty$ (or $\alpha\to\infty$). In other words, the Power Watershed counts by how many seed separating mSFs a node is connected to a seed (see \figurename{} \ref{fig:Count_MSF}). Note, that there can be more than one mSF when the edge-costs are not unique. In \figurename{} \ref{fig:edge_presence} we show the probability of an edge being part of a mSF (see Appendix \ref{App:Edge_prob} for a more exhaustive explanation). In addition, it is worth recalling that the cut given by the Power Watershed segmentation is a mSF-cut (Property 2 of \cite{Couprie2011}).

The Random Walker interpretation can break down in the limit case of the Power Watershed. After taking the power of the edge-weights to infinity, at any node a Random Walker would move along an incident edge with maximum weight / minimum cost. So, in the limit case a Random Walker could get stuck at the edges, $e = \{u,v\}$, which minimize the cost among all the edges incident to $u$ or $v$. In this case the Random Walker will not necessarily reach any seed (see \figurename{} \ref{sfig6:Count_MSF}). In contrast, the forest-counting interpretation carries over nicely to the limit case.

The Probabilistic Watershed with a Gibbs distribution over 2-forests of minimal (maximal) entropy, $\mu = \infty$ ($\mu = 0$), corresponds to the Power Watershed (only considers the graph's topology). The effect of $\mu$ is illustrated on a a toy graph in \figurename{} \ref{fig:mu_behaviour} of the appendix. One could perform grid search to identify interesting intermediate values of $\mu$. Alternatively, $\mu$ can be learned, alongside the edge-costs, by back-propagation \cite{Cerrone_RW} or by a first-order approximation thereof \cite{vernaza2017learning}.

\section{Discussion}
\label{Discussion}
In this work, we provided new understanding of well-known seeded segmentation algorithms.

We have presented a tractable way of computing the expected label assignment of each node by a Gibbs distribution over all the seed separating spanning forests of a graph (\thref{Def:Prob}). Using the MTT we showed that this is computationally and by result equivalent to the Random Walker \cite{Grady2006}. Our approach has been developed without using potential theory (in contrast to \cite{Biggs1997}).

These facts have provided us with a novel understanding of the Random Walker (Probabilistic Watershed) probabilities: They are proportional to the gap produced by the triangle inequality of the effective resistance between the seeds and the query node.

Finally, we have proposed a new interpretation of the Power Watershed potentials for $\beta=2$ and $\alpha\to\infty$: They are given as the probabilities of the Probabilistic Watershed when the latter is restricted to mSFs instead of all spanning forests.

 A mSF can also be seen as a union of minimax paths between the vertices \cite{MAGGS1988}. Recently, \cite{najman2019} showed that the Power Watershed assigns a query node $q$ to the seed to which the minimax path from $q$ has the lowest maximum edge cost. In future work, we hope to extend this path-related point of view to an intuitive understanding of the Power Watershed.
 
 We are currently working on an extension of the Probabilistic Watershed framework to directed graphs, by means of the generalization of the MTT to directed graphs \cite{Tutte1984}. Here, one samples directed spanning forests with the seeds as sinks to segment the unlabelled nodes. This might lead to a new practical algorithm for semi-supervised learning on directed graphs such as social / citation or Web networks and could be related to directed random walks.

\section*{Acknowledgements}

The authors would like to thank Prof. Marco Saerens for his profound and constructive comments as well as the anonymous reviewers for their helpful remarks.  We would like to express our gratitude to Lorenzo Cerrone, who also shared the edge weights of [11], and Laurent Najman for the useful discussions about the Random Walker and Power Watershed algorithms, respectively.  We also acknowledge partial financial support of the DFG under grant No. DFG HA-4364 8-1.

\bibliographystyle{abbrv}
\bibliography{Bibliography_Probabilistic_Watershed}

\begin{thebibliography}{10}

\bibitem{Allene2007}
C.~All{\`{e}}ne, J.~Y. Audibert, M.~Couprie, J.~Cousty, and R.~Keriven.
\newblock {Some links between min-cuts, optimal spanning forests and
  watersheds}.
\newblock In {\em Proceedings of the 8th International Symposium on
  Mathematical Morphology}, pages 253--264, 2007.

\bibitem{stochasticwatershed}
J.~Angulo and D.~Jeulin.
\newblock Stochastic watershed segmentation.
\newblock In {\em Proceedings of the 8th International Symposium on
  Mathematical Morphology}, pages 265--276, 2007.

\bibitem{DeepWatershed}
M.~Bai and R.~Urtasun.
\newblock Deep watershed transform for instance segmentation.
\newblock {\em CVPR}, pages 2858--2866, 2017.

\bibitem{beier2017multicut}
T.~Beier, C.~Pape, N.~Rahaman, T.~Prange, S.~Berg, D.~D. Bock, A.~Cardona,
  G.~W. Knott, S.~M. Plaza, L.~K. Scheffer, U.~Koethe, A.~Kreshuk, and F.~A.
  Hamprecht.
\newblock Multicut brings automated neurite segmentation closer to human
  performance.
\newblock {\em Nature Methods}, 14(2):101—102, January 2017.

\bibitem{belkin2006manifold}
M.~Belkin, P.~Niyogi, and V.~Sindhwani.
\newblock Manifold regularization: A geometric framework for learning from
  labeled and unlabeled examples.
\newblock {\em JMLR}, 7:2399--2434, 2006.

\bibitem{beucher1979watersheds}
S.~Beucher and C.~Lantuéjoul.
\newblock Use of watersheds in contour detection.
\newblock In {\em International Workshop on Image Processing: Real-time Edge
  and Motion Detection/Estimation}, volume 132, 1979.

\bibitem{1993Beucher}
S.~Beucher and F.~Meyer.
\newblock The morphological approach to segmentation: the watershed
  transformation.
\newblock {\em Mathematical Morphology in Image Processing}, 34:433--481, 1993.

\bibitem{Biggs1997}
N.~Biggs.
\newblock Algebraic potential theory on graphs.
\newblock {\em Bulletin of the London Mathematical Society}, 29:641--682, 1997.

\bibitem{bockelmann2019}
N.~Bockelmann, D.~Kr{\"u}ger, D.~F. Wieland, B.~Zeller-Plumhoff, N.~Peruzzi,
  S.~Galli, R.~Willumeit-R{\"o}mer, F.~Wilde, F.~Beckmann, J.~Hammel, et~al.
\newblock Sparse annotations with random walks for u-net segmentation of
  biodegradable bone implants in synchrotron microtomograms.
\newblock In {\em International Conference on Medical Imaging with Deep
  Learning -- Extended Abstract Track}, 2019.

\bibitem{MTT_eigenvalues}
A.~E. Brouwer and W.~H. Haemers.
\newblock {\em Spectra of Graphs}.
\newblock Springer, New York, NY, 2012.

\bibitem{Bui2018}
V.~Bui, L.-Y. Hsu, L.-C. Chang, and M.~Y. Chen.
\newblock An automatic random walk based method for 3{D} segmentation of the
  heart in cardiac computed tomography images.
\newblock In {\em ISBI}, pages 1352--1355, 2018.

\bibitem{Cerrone_RW}
L.~Cerrone, A.~Zeilmann, and F.~A. Hamprecht.
\newblock End-to-end learned random walker for seeded image segmentation.
\newblock In {\em CVPR}, 2019.

\bibitem{najman2019}
A.~{Challa}, S.~{Danda}, B.~S.~D. {Sagar}, and L.~{Najman}.
\newblock Watersheds for semi-supervised classification.
\newblock {\em IEEE Signal Processing Letters}, 26:720--724, May 2019.

\bibitem{Chazelle2000}
B.~Chazelle.
\newblock A minimum spanning tree algorithm with inverse-ackermann type
  complexity.
\newblock {\em J. ACM}, 47:1028--1047, November 2000.

\bibitem{Forest_measure}
P.~Y. Chebotarev and E.~Shamis.
\newblock The matrix-forest theorem and measuring relations in small social
  groups1.
\newblock {\em Automation and Remote Control}, 58:1505--1514, 1997.

\bibitem{Couprie2011}
C.~Couprie, L.~Grady, L.~Najman, and H.~Talbot.
\newblock Power watershed: A unifying graph-based optimization framework.
\newblock {\em IEEE Transactions on Pattern Analysis and Machine Intelligence},
  33:1384--1399, 2011.

\bibitem{Topological_Watershed}
M.~Couprie, L.~Najman, and G.~Bertrand.
\newblock Quasi-linear algorithms for the topological watershed.
\newblock {\em Journal of Mathematical Imaging and Vision}, 22:231--249, 2005.

\bibitem{Cousty2009}
J.~Cousty, G.~Bertrand, L.~Najman, and M.~Couprie.
\newblock Watershed cuts: Minimum spanning forests and the drop of water
  principle.
\newblock {\em IEEE Transactions on Pattern Analysis and Machine Intelligence},
  31:1362--74, 2009.

\bibitem{CREMI}
CREMI.
\newblock Miccai challenge on circuit reconstruction from electron microscopy
  images, 2017.
\newblock \url{https://cremi.org}.

\bibitem{criminisi2008geos}
A.~Criminisi, T.~Sharp, and A.~Blake.
\newblock Geos: Geodesic image segmentation.
\newblock In {\em ECCV}, pages 99--112. Springer, 2008.

\bibitem{cuturi2013sinkhorn}
M.~Cuturi.
\newblock Sinkhorn distances: Lightspeed computation of optimal transport.
\newblock In {\em NIPS}, pages 2292--2300, 2013.

\bibitem{Fernandes2019}
S.~E.~N. Fernandes and J.~P. Papa.
\newblock Improving optimum-path forest learning using bag-of-classifiers and
  confidence measures.
\newblock {\em Pattern Analysis and Applications}, 22:703--716, 2019.

\bibitem{Fernandes2019Prob_OPF}
S.~E.~N. {Fernandes}, D.~R. {Pereira}, C.~C.~O. {Ramos}, A.~N. {Souza}, D.~S.
  {Gastaldello}, and J.~P. {Papa}.
\newblock A probabilistic optimum-path forest classifier for non-technical
  losses detection.
\newblock {\em IEEE Transactions on Smart Grid}, 10:3226--3235, 2019.

\bibitem{Fouss2016}
F.~Fouss, M.~Saerens, and M.~Shimbo.
\newblock {\em Algorithms and Models for Network Data and Link Analysis}.
\newblock Cambridge University Press, New York, NY, USA, 1st edition, 2016.

\bibitem{franccoisse2017bag}
K.~Fran{\c{c}}oisse, I.~Kivim{\"a}ki, A.~Mantrach, F.~Rossi, and M.~Saerens.
\newblock A bag-of-paths framework for network data analysis.
\newblock {\em Neural Networks}, 90:90--111, 2017.

\bibitem{Shayan_spanning_tree}
S.~O. Gharan.
\newblock Recent developments in approximation algorithms: Lecture 1 and 2:
  Random spanning trees.
\newblock
  \url{https://homes.cs.washington.edu/~shayan/courses/cse599/index.html},
  2015.

\bibitem{Ghosh2008}
A.~Ghosh, S.~Boyd, and A.~Saberi.
\newblock Minimizing effective resistance of a graph.
\newblock {\em SIAM Rev.}, 50:37--66, 2008.

\bibitem{Grady2006}
L.~Grady.
\newblock {Random walks for image segmentation}.
\newblock {\em IEEE Transactions on Pattern Analysis and Machine Intelligence},
  2006.

\bibitem{gradytextbook}
L.~J. Grady and J.~R. Polimeni.
\newblock {\em {Discrete Calculus}}.
\newblock Springer, London, 2010.

\bibitem{Harville2008_Book}
D.~A. Harville.
\newblock {\em Matrix Algebra From a Statistician's Perspective}.
\newblock Springer-Verlag New York, 1997.

\bibitem{kim2014label}
K.-H. Kim and S.~Choi.
\newblock Label propagation through minimax paths for scalable semi-supervised
  learning.
\newblock {\em Pattern Recognition Letters}, 45:17--25, 2014.

\bibitem{kirchhoff1847ueber}
G.~Kirchhoff.
\newblock {\"U}ber die {A}ufl{\"o}sung der {G}leichungen, auf welche man bei
  der {U}ntersuchung der linearen {V}ertheilung galvanischer {S}tr{\"o}me
  gef{\"u}hrt wird.
\newblock {\em Annalen der Physik}, 148:497--508, 1847.

\bibitem{DevelopmentsRSP}
I.~Kivim{\"{a}}ki, M.~Shimbo, and M.~Saerens.
\newblock {Developments in the theory of randomized shortest paths with a
  comparison of graph node distances}.
\newblock {\em Physica A: Statistical Mechanics and its Applications},
  393:600--616, 2014.

\bibitem{koo2007structured}
T.~K. Koo, A.~Globerson, X.~Carreras, and M.~Collins.
\newblock Structured prediction models via the matrix-tree theorem.
\newblock In {\em EMNLP-CoNLL}, 2007.

\bibitem{Kruskal1956}
J.~B. Kruskal.
\newblock {On the Shortest Spanning Subtree of a Graph and the Traveling
  Salesman Problem}.
\newblock In {\em Proceedings of the American Mathematical Society, 7}, 1956.

\bibitem{Liu2018}
Z.~Liu, Y.~Song, C.~Maere, Q.~Liu, Y.~Zhu, H.~Lu, and D.~Yuan.
\newblock A method for {PET-CT} lung cancer segmentation based on improved
  random walk.
\newblock {\em 24th International Conference on Pattern Recognition (ICPR)},
  pages 1187--1192, 2018.

\bibitem{MAGGS1988}
B.~M. Maggs and S.~A. Plotkin.
\newblock Minimum-cost spanning tree as a path-finding problem.
\newblock {\em Information Processing Letters}, 26:291 -- 293, 1988.

\bibitem{MALMBERG2014}
F.~Malmberg and C.~L.~L. Hendriks.
\newblock An efficient algorithm for exact evaluation of stochastic watersheds.
\newblock {\em Pattern Recognition Letters}, 47:80 -- 84, 2014.
\newblock Advances in Mathematical Morphology.

\bibitem{mensch2018differentiable}
A.~Mensch and M.~Blondel.
\newblock Differentiable dynamic programming for structured prediction and
  attention.
\newblock In {\em ICML}, 2018.

\bibitem{ofverstedt2019stochastic}
J.~{\"O}fverstedt, J.~Lindblad, and N.~Sladoje.
\newblock Stochastic distance transform.
\newblock In {\em International Conference on Discrete Geometry for Computer
  Imagery}, pages 75--86. Springer, 2019.

\bibitem{Rao_Mitra}
C.~Rao and S.~Mitra.
\newblock {\em Generalized Inverse of Matrices and Its Applications}.
\newblock John Wiley and Sons, 1971.

\bibitem{Senelle2014}
M.~Senelle, S.~Garc\'\i{a}-D\'\i{e}z, A.~Mantrach, M.~Shimbo, M.~Saerens, and
  F.~Fouss.
\newblock The sum-over-forests density index: {I}dentifying dense regions in a
  graph.
\newblock {\em IEEE Transactions on Pattern Analysis and Machine Intelligence},
  36:1268--1274, 2014.

\bibitem{Watershed_uncertainty_stimators}
C.~{Straehle}, U.~{Koethe}, G.~{Knott}, K.~{Briggman}, W.~{Denk}, and F.~A.
  {Hamprecht}.
\newblock Seeded watershed cut uncertainty estimators for guided interactive
  segmentation.
\newblock In {\em CVPR}, pages 765--772, 2012.

\bibitem{Teixeira2013}
A.~Teixeira, P.~Monteiro, J.~Carriço, M.~Ramirez, and A.~Francisco.
\newblock Spanning edge betweenness.
\newblock In {\em Workshop on Mining and Learning with Graphs}, volume~24,
  pages 27--31, January 2013.

\bibitem{Teixeira2015}
A.~Teixeira, P.~Monteiro, J.~Carriço, M.~Ramirez, and A.~P.~Francisco.
\newblock Not seeing the forest for the trees: Size of the minimum spanning
  trees (msts) forest and branch significance in mst-based phylogenetic
  analysis.
\newblock {\em PLOS ONE}, 10, 2015.

\bibitem{Tsen1994}
F.-S. Tsen, T.-Y. Sung, M.-Y. Lin, L.-H. Hsu, and W.~Myrvold.
\newblock Finding the most vital edge with respect to the number of spanning
  trees.
\newblock {\em IEEE Transactions on Reliability}, 43:600--602, 1994.

\bibitem{Tutte1984}
W.~T. Tutte.
\newblock {\em Graph theory}.
\newblock Encyclopedia of mathematics and its applications. Addison-Wesley,
  1984.

\bibitem{tzeng2000spanning}
W.-J. Tzeng and F.~Wu.
\newblock Spanning trees on hypercubic lattices and nonorientable surfaces.
\newblock {\em Applied Mathematics Letters}, 13(7):19--25, 2000.

\bibitem{vernaza2017learning}
P.~Vernaza and M.~Chandraker.
\newblock Learning random-walk label propagation for weakly-supervised semantic
  segmentation.
\newblock In {\em CVPR}, pages 7158--7166, 2017.

\bibitem{winkler2012image}
G.~Winkler.
\newblock {\em Image analysis, random fields and Markov chain Monte Carlo
  methods: a mathematical introduction}, volume~27.
\newblock Springer Science \& Business Media, 2012.

\bibitem{Learned_Watershed}
S.~Wolf, L.~Schott, U.~K{\"o}the, and F.~A. Hamprecht.
\newblock Learned watershed: End-to-end learning of seeded segmentation.
\newblock {\em ICCV}, pages 2030--2038, 2017.

\bibitem{zhou2004learning}
D.~Zhou, O.~Bousquet, T.~N. Lal, J.~Weston, and B.~Sch{\"o}lkopf.
\newblock Learning with local and global consistency.
\newblock In {\em NIPS}, pages 321--328, 2004.

\bibitem{Zhu2003}
X.~Zhu, Z.~Ghahramani, and J.~Lafferty.
\newblock Semi-supervised learning using gaussian fields and harmonic
  functions.
\newblock In {\em ICML}, pages 912--919, 2003.

\end{thebibliography}

 \appendix
\section{Calculus of $w\left(\Forest[u][v]\right)$}
\label{app:calculus_F}
In this appendix we will focus on the calculation of $w(\Forest[u][v])$ for any nodes $u,v\in V$. We will prove \thref{lem:T_e_weighted} (here denoted \thref{lem_app:T_e_weighted}). In order to demonstrate it we will use the matrix tree theorem (MTT) and some previous results. 

To deduce some results we will use multigraphs. The Laplacian of a multigraph is slightly different from the Laplacian of a simple graph. We introduce the following definition.
\begin{definition}[\textbf{Laplacian of a multigraph}]
	Let $G$ be a multigraph.  The Laplacian $L_G$ has the following formula
	\[\big(L_G\big)_{uv}:=\begin{cases}
	\sum_{\bar{e}\in E^{(u,v)}}-w_G\big(\bar{e}\big) &\text{ if } u\neq v\\
	\sum_{\substack{k\in V \\ k\neq u}} \sum_{\bar{e}\in E^{(u,k)}}w_G\big(\bar{e}\big) &\text{ if } u=v
	\end{cases},\]
	where $E^{(u,v)}\subset E_G$ is the subset of edges incident to $u$ and $v$. If there are no edges incident to $u$ and $v$ the sum is considered equal to $0$.
\end{definition}

Since the MTT is crucial in our theory we recall it.

\begin{theorem}[\textbf{Matrix tree theorem (MTT)}]
	\thlabel{Th_app:Matrix_tree}
	For any weighted multigraph $G$ the sum of the weights of the spanning trees of $G$, $w(\Setspt):=\sum_{t\in \Setspt }\prod_{e\in E_t}w(e)$, is equal to
	\[w(\Setspt)=\det(L^{[v]})=\frac{1}{|V|}\lambda_2\cdots\lambda_{|V|}=\frac{1}{|V|}\det(L+\frac{1}{|V|}\mathbbm{1}\mathbbm{1}^\top),\]
	where  $L^{[v]}$ is the matrix obtained from the Laplacian, $L$,  after removing the row and column corresponding to an arbitrary but fixed node $v$, $\{\lambda_i\}_{i\geq 1}$ are the eigenvalues of L with $\lambda_1=0$ and $\mathbbm{1}$ the corresponding eigenvector, the column vector of $1$'s.
\end{theorem}

\begin{proof}
We will only prove the third equality. The first equality can be found in \cite{Tutte1984}. The second equality can be found in chapter 1 of \cite{MTT_eigenvalues} for non-weighted graphs, but the reasoning of the proof is equivalent for weighted graphs. The proof of the third equality is Theorem 1.6 in \cite{Shayan_spanning_tree}, but since it has some typos we preferred to prove it by means of the second equality.

In order to demonstrate the third equality we will show that 
\begin{equation}
\label{eq:key_equality_MTT_app_proof}
\lambda_2\cdots\lambda_{|V|}=\det(L+\frac{1}{|V|}\mathbbm{1}\mathbbm{1}^\top).
\end{equation}

 Let $\tilde{\lambda}_i$ for  $i=1,\dots,|V|$ be the eigenvalues of $L+\frac{1}{|V|}\mathbbm{1}\mathbbm{1}^\top$. Since the determinant of a matrix is the product of its eigenvalues we obtain

\[\det(L+\frac{1}{|V|}\mathbbm{1}\mathbbm{1}^\top)=\tilde{\lambda}_1\cdots \tilde{\lambda}_{|V|}.\]

We will show that one of the $\tilde{\lambda}_i$'s, say $\tilde{\lambda}_1$, is one and that $\{\tilde{\lambda}_2 ,\dots, \tilde{\lambda}_{|V|}\} = \{\lambda_2,\dots,\lambda_{|V|}\}$ which establishes equation  \eqref{eq:key_equality_MTT_app_proof}.

The first eigenvalue of the Laplacian $L$ is $\lambda_1=0$ whose eigenvector is $\mathbbm{1}$, since the elements of every row of $L$ sum to 1. We prove now that $\mathbbm{1}$ is an eigenvector of  $L+\frac{1}{|V|}\mathbbm{1}\mathbbm{1}^\top$ with eigenvalue equal to $1$.

\[(L+\frac{1}{|V|}\mathbbm{1}\mathbbm{1}^\top)\mathbbm{1}=\underbrace{L\mathbbm{1}}_{=0}+\frac{1}{|V|}\mathbbm{1}\underbrace{\mathbbm{1}^\top\mathbbm{1}}_{=|V|}=\mathbbm{1}\]

Therefore, we get $\tilde{\lambda}_1=1$. Since $L+\frac{1}{|V|}\mathbbm{1}\mathbbm{1}^\top$ is symmetric, we can find an orthogonal basis of eigenvectors of $L+\frac{1}{|V|}\mathbbm{1}\mathbbm{1}^\top$ containing $\mathbbm{1}$. Let $x_i$ be an element of that basis associated with $\lambda_i$ for $i\geq 2$.  By the orthogonality of $\mathbbm{1}$ and $x_i$, we get
\[(L+\frac{1}{|V|}\mathbbm{1}\mathbbm{1}^\top)x_i=Lx_i+\frac{1}{|V|}\mathbbm{1}\underbrace{\mathbbm{1}^\top x_i}_{=0}=\lambda_i x_i\]

Therefore $\lambda_i$ for $i \geq 2$ is also an eigenvalue of $L+\frac{1}{|V|}\mathbbm{1}\mathbbm{1}^\top$ and the theorem is proven.
\end{proof}

\begin{lemma}[Determinant Lemma]
	
	\thlabel{lem:Determinant_lemma}
	Given an invertible matrix $A\in \mathbb{R}^{m\times m}$ and $u,v\in \mathbb{R}^m$ then:
	\[\det(A+uv^\top)=\det(A)(1+v^\top A^{-1} u)\]
\end{lemma}
\begin{proof} See \cite{Harville2008_Book}.
\end{proof}
\begin{lemma}
	\thlabel{lem:L^+}
	Let $L^+$ be the pseudo-inverse of the Laplacian, then
	
	\[L^+=\left(L+\frac{\mathbbm{1}\mathbbm{1}^\top}{|V|}\right)^{-1}-\frac{\mathbbm{1}\mathbbm{1}^\top}{|V|},\]
	where $\mathbbm{1}$ is the column vector of $1$s.
\end{lemma}
\begin{proof}

	On page 48 of \cite{Fouss2016} the same argument as in the proof of \thref{Th_app:Matrix_tree} is applied. $\left(L+\frac{\mathbbm{1}\mathbbm{1}^\top}{|V|}\right)^{-1}$ and $L^+$ have the same eigenvectors and eigenvalues except for $\tilde{\lambda}_1=1$ whose corresponding eigenvalue of $L^+$  is $\lambda_1=0$. By subtracting $\frac{\mathbbm{1}\mathbbm{1}^\top}{|V|}$ from $\left(L+\frac{\mathbbm{1}\mathbbm{1}^\top}{|V|}\right)^{-1}$ the eigenvalue $\tilde{\lambda}_1$ is modified and becomes $0$. Hence both matrices are equal since they have the same spectral decomposition . See chapter 10 in \cite{Rao_Mitra} for more details.
\end{proof}

Before proving the next result we need to introduce some notation. Let $\bar{e} = \{u,v\}\subset V$ be an edge not necessarily included in $E$. Let $G_{\bar{e}}$ be the graph formed from $G$ after adding the edge $\bar{e}$, i.e. $G_{\bar{e}}=(V,E\sqcup \{\bar{e}\})$, where $\sqcup$ denotes the disjoint union\footnote{Given a family of sets $\{A_i : i\in I\}$ the disjoint union is defined as $\bigsqcup_{i\in I}A_i=\bigcup_{i\in I}\left\{(x,i)\ : \ x\in A_i \right\}$.}. The use of the disjoint union permits distinguishing between the added edge and the ones originally present in the graph. Therefore $G_{\bar{e}}$ may be a multigraph. The following lemma explains why we need to consider $G_{\bar{e}}$ as a multigraph.

\begin{lemma}
	\thlabel{lem_app_w(Te)+w(T)=w(Ge)}
Let $G=(V_G,E_G,w_G)$  be a weighted graph and consider $G_{\bar{e}}=(V_G,E_G\sqcup\{\bar{e}\},w)$ for some $\bar{e} = \{u,v\}\subset V$, where $w(e)=w_G(e)$ for all $e\in E_G$ and $w(\bar{e})$ an arbitrary positive number. We will omit the subscript of $w_G$ without risk of confusion. Then
\[w\left(\Forest[u][v]\right)=\frac{w(\Setspt[G_{\bar{e}}])-w(\Setspt[G])}{w(\bar{e})},\]
where $\Setspt[G_{\bar{e}}]$ and $\Setspt[G]$ denote the set of spanning trees of $G_{\bar{e}}$ and $G$ respectively.
\end{lemma}
\begin{proof}
The key idea of the proof is the fact that $\Setspt[G_{\bar{e}}]$ can be partitioned in the set of trees that do not contain the edge $\bar{e}$ (which is equal to $\Setspt[G]$ since $G$ did not contain $\bar{e}$) and the ones that do contain the edge $\bar{e}$. Let us denote the second set $\SetsptE[\bar{e}]:=\{t\in\Setspt[G_{\bar{e}}] \ : \ \bar{e}\in E_t\}$. Recall that $\bar{e}$ is considered as a special edge even if there was already an edge $e=\{u,v\}$ contained in the graph $G$. If $e$ and $\bar{e}$ where considered as the same edge, then the set of trees not containing $\bar{e}$ would not be equal to $\Setspt[G]$, since $e$ belongs to some trees in $\Setspt[G]$. Therefore we need to consider $G_{\bar{e}}$ as a multigraph.

Note that there is a bijection between $\SetsptE[\bar{e}]$ and $\mathcal{F}^{v}_{u}$ since any tree $t \in \SetsptE[\bar{e}]$ forms a 2-forest $f\in \mathcal{F}^{v}_{u}$ after removing $\bar{e}$ from $t$, and vice versa, any $f\in \mathcal{F}^{v}_{u}$ forms a tree in $t \in \SetsptE[\bar{e}]$ after adding $\bar{e}$ (see p.652 in \cite{Biggs1997}). Moreover, $w(\bar{e})\cdot w(f)=w(t)$ since the only edge present in $t$ but not in $f$ is $\bar{e}$. Therefore, we obtain
\begin{equation}
\begin{aligned}
	w(\Setspt[G_{\bar{e}}])&=\sum_{t\in \mathcal{T}_{G_{\bar{e}}} }w(t)=\sum_{t\in\SetsptE[\bar{e}]}w(t) + \sum_{t\in\Setspt[G]}w(t)=w(\bar{e})\sum_{t\in\SetsptE[\bar{e}]}\frac{w(t)}{w(\bar{e})} + w(\Setspt[G])\\
	&=w(\bar{e})\sum_{f\in \Forest[u][v] }w(f)+w(\Setspt[G])=w(\bar{e}) w(\Forest[u][v])+w(\Setspt[G]).
\end{aligned}
\end{equation}
Isolating $w(\Forest[u][v])$ we get the desired result.
\end{proof}

\begin{lemma}[\textbf{Lemma 2.2}]	
	\thlabel{lem_app:T_e_weighted}
	Let $G=(V,E,w)$  be an undirected edge-weighted connected graph and $u,v\in V$ arbitrary vertices.
	\begin{enumerate}[leftmargin=*]
		\item[(a)] Let $\ell_{ij}^+$ denote the entry $ij$ of the  pseudo-inverse of the Laplacian of $G$, $L^+_G$,  then
		\begin{equation}
		w(\Forest[u][v])=w(\Setspt)\left(\ell^+_{uu}+\ell^+_{vv}-2\ell^+_{uv}\right).
		\label{eq:T_e_weighted_pseudoinverse}
		\end{equation}
		
		\item[(b)] If $\ell_{ij}^{-1,[r]}$ denotes the entry $ij$ of the inverse of the matrix $L^{[r]}$ (the Laplacian $L$ after removing the row and the column corresponding to node $r$), then
		\begin{equation}
		w(\Forest[u][v])=\begin{cases}
		w(\Setspt)\left(\ell_{uu}^{-1,[r]}+\ell_{vv}^{-1,[r]}-2\ell_{uv}^{-1,[r]}\right) &\text{ if } r\neq u,v\\
		w(\Setspt)\ell_{uu}^{-1,[v]} &\text{ if } r=v \text{ and } u\neq v \\
		w(\Setspt)\ell_{vv}^{-1,[u]} &\text{ if } r=u \text{ and } u\neq v. \\
		\end{cases}
		\label{eq:app_T_e_weighted_remove_r}
		\end{equation}
	\end{enumerate}
\end{lemma}
\begin{proof}~
	We apply the matrix tree theorem (\thref{Th_app:Matrix_tree}) in combination with \thref{lem:Determinant_lemma}, \thref{lem:L^+} and \thref{lem_app_w(Te)+w(T)=w(Ge)}. We will use the following notation 
	\[\tilde{L}_G=L_G+\frac{\mathbbm{1}\mathbbm{1}^\top}{|V|}.\]
	
	In order to use \thref{lem_app_w(Te)+w(T)=w(Ge)} we will use the edge $\bar{e}=\{u,v\}$ with $w(\bar{e})=1$.	Moreover, let us denote $b_{\bar{e}}=\textbf{1}_u-\textbf{1}_v$ where $\textbf{1}_v$ indicates the column $v$ of the identity matrix. Since the difference between $G$ and $G_{\bar{e}}$ is just the edge $\bar{e}$ with $w(\bar{e})=1$, we can write the following relation between the Laplacians of $G$ and $G_{\bar{e}}$
	\[L_{G_{\bar{e}}}=L_G+b_{\bar{e}}b_{\bar{e}}^\top.\]
	Therefore, we may write
	\begin{equation*}
	\begin{aligned}
w(\Forest[u][v])&\underbrace{=}_{\text{\thref{lem_app_w(Te)+w(T)=w(Ge)}}}	w(\Setspt[G_{\bar{e}}])-w(\Setspt[G])\underbrace{=}_{\substack{\text{\textbf{MTT}}\\ \text{\thref{Th_app:Matrix_tree}}}}\frac{1}{|V|}\det\Big(\underbrace{\tilde{L}_G+b_eb_e^\top}_{\tilde{L}_{G_{\bar{e}}}}\Big)-\frac{1}{|V|}\det(\tilde{L}_G)\\
	&\underbrace{=}_{\text{\thref{lem:Determinant_lemma}}}\frac{1}{|V|}\det(\tilde{L}_G)\left(1+b_{\bar{e}}^\top\tilde{L}_G^{-1}b_{\bar{e}}\right)-\frac{1}{|V|}\det(\tilde{L}_G)=\frac{1}{|V|}\det(\tilde{L}_G)\left(b_{\bar{e}}^\top\tilde{L}_G^{-1}b_{\bar{e}}\right)\\
	&\underbrace{=}_{\text{\thref{lem:L^+}}}\frac{1}{|V|}\det(\tilde{L}_G)\left(b_{\bar{e}}^\top\left(L^++\frac{\mathbbm{1}\mathbbm{1}^\top}{|V|}\right)b_{\bar{e}}\right)=
	\frac{1}{|V|}\det(\tilde{L}_G)\left(\ell^+_{uu}+\ell^+_{vv}-2\ell^+_{uv}\right)\\
	&\underbrace{=}_{\substack{\text{\textbf{MTT}}\\ \text{\thref{Th_app:Matrix_tree}}}}w(\Setspt[G])\left(\ell^+_{uu}+\ell^+_{vv}-2\ell^+_{uv}\right)
	\end{aligned}
	\end{equation*}
	The second statement can be deduced from equation \eqref{eq:Effective_resistance}
	in the main paper. Still we show how it can be computed by following a similar argument as in the previous case. Let $r\neq u,v$.
	\begin{equation*}
	\begin{aligned}
	w(\Forest[u][v])&\underbrace{=}_{\text{\thref{lem_app_w(Te)+w(T)=w(Ge)}}}	w(\Setspt[G_{\bar{e}}])-w(\Setspt[G])\underbrace{=}_{\substack{\text{\textbf{MTT}}\\ \text{\thref{Th_app:Matrix_tree}}}}\det\bigg(\underbrace{L^{[r]}_G+b_e^{[r]}\left(b_e^{[r]}\right)^{\top}}_{L^{[r]}_{G_{\bar{e}}}}\bigg)-\det(L^{[r]}_G)\\
	&\underbrace{=}_{\text{\thref{lem:Determinant_lemma}}} \det(L^{[r]}_G)\left(1+\left(b_{\bar{e}}^{[r]}\right)^{\top} \left(L_G^{[r]}\right)^{-1}b_{\bar{e}}^{[r]}\right)-\det(L^{[r]}_G)\\
	&\quad\,\,=\det(L^{[r]}_G)\left(\ell_{uu}^{-1,[r]}+\ell_{vv}^{-1,[r]}-2\ell_{uv}^{-1,[r]}\right) =w(\Setspt[G])\left(\ell_{uu}^{-1,[r]}+\ell_{vv}^{-1,[r]}-2\ell_{uv}^{-1,[r]}\right).
	\end{aligned}
	\end{equation*}

	For $r=u$ the proof is the following:
	\begin{equation*}
	\begin{aligned}
	w(\Forest[u][v]) &\underbrace{=}_{\text{\thref{lem_app_w(Te)+w(T)=w(Ge)}}} w(\Setspt[G_{\bar{e}}])-w(\Setspt[G])\underbrace{=}_{\substack{\text{\textbf{MTT}}\\ \text{\thref{Th_app:Matrix_tree}}}}\det\bigg(\underbrace{L^{[u]}_G+\textbf{1}_v^{[u]}\left(\textbf{1}_v^{[u]}\right)^\top}_{L^{[u]}_{G_e}}\bigg)-\det(L^{[u]}_G)\\
	&\underbrace{=}_{\text{\thref{lem:Determinant_lemma}}} \det(L^{[u]}_G)\left(1+\left(\textbf{1}_v^{[u]}\right)^\top \left(L_G^{[u]}\right)^{-1}\textbf{1}_v^{[u]}\right)-\det(L^{[u]}_G)=\det(L^{[u]}_G)\ell_{vv}^{-1,[u]}\\
	&\quad\,\,=w(\Setspt[G])\ell_{vv}^{-1,[u]}.		
	\end{aligned}
	\end{equation*}
	
	The case $r=v$ is analogous.
\end{proof}

\section{Proof of \thref{theorem:Trees_equiv_Random_walk_2_seeds}}
\label{App:Proof_ProbWS=RW}
\begin{theorem}[\thref{theorem:Trees_equiv_Random_walk_2_seeds}]
	\thlabel{App_theorem:Trees_equiv_Random_walk_2_seeds} 
	The probability $x_{q}^{s_1}$ that a random walker as defined in \cite{Grady2006} starting at node $q$ reaches $s_1$ first before reaching $s_2$ is equal to the Probabilistic Watershed probability defined in Definition 3.1 of the main paper:
	\[x_{q}^{s_1}=P(q\sim s_1).\]
\end{theorem}
\begin{proof}
	If we write the probability in terms of the inverse of $L^{[s_2]}$ (Lemma 2.2, equation 2 of the main paper) we find:
	\begin{equation}
	\label{eq:P(q-s2)_terms[s1]}
	\begin{aligned}
	&P(q \sim s_1)=\big(w(\Forest[s_2][q])+w(\Forest[s_1][s_2])-w(\Forest[s_1][q])\big)\big/\big(2w(\Forest[s_1][s_2])\big)\\
	&=\left(\ell_{qq}^{ -1,[s_2]}+\ell_{s_1s_1}^{ -1,[s_2]}-\left(\ell_{s_1s_1}^{ -1,[s_2]}+\ell_{qq}^{ -1,[s_2]}-2\ell_{qs_1}^{ -1,[s_2]}\right)\right)\big/\left(2\ell_{s_1s_1}^{ -1,[s_2]}\right)=\ell_{qs_1}^{ -1,[s_2]}\big/\ell_{s_1s_1}^{ -1,[s_2]}.
	\end{aligned}
	\end{equation}
	Therefore, to calculate the probabilities for $P(q \sim s_1)$ for every $q$ we only need to compute the column $s_1$ of $\left(L^{[s_2]}\right)^{-1}$. Solving the following linear system:
	\begin{equation}
	\label{eq:L-1_seed_prob}
	L^{[s_2]}y=\textbf{1}_{s_1}\big/\ell_{s_1s_1}^{ -1,[s_2]}\Longleftrightarrow y=\left(L^{[s_2]}\right)^{-1}\textbf{1}_{s_1}\big/\ell_{s_1s_1}^{ -1,[s_2]}=\left(L^{[s_2]}\right)^{-1}_{\cdot,s_1}\big/\ell_{s_1s_1}^{ -1,[s_2]},
	\end{equation}
	where $\textbf{1}_{u}$ denotes the column $u$ of the identity matrix, we have that $y$ is the vector formed by the elements in the right hand side of \eqref{eq:P(q-s2)_terms[s1]}. 
	Let us assume without loss of generality that the row corresponding to the seed $s_1$ is the first one, then we can express equation \eqref{eq:L-1_seed_prob} block-wise :
	\begin{equation}
	\left(\begin{array}{cc}
	L_{s_1s_1}&B_{s_1}^{\top}\\
	B_{s_1} &L_U
	\end{array}\right)\left(\begin{array}{c}
	y_{s_1}\\y_U
	\end{array}\right)=
	\left(\begin{array}{c}
	L_{s_1s_1}y_{s_1}+B_{s_1}^{\top}y_U\\
	B_{s_1}y_{s_1} +L_Uy_U
	\end{array}\right)=\left(\begin{array}{c}
	1/\ell_{s_1s_1}^{ -1,[s_2]}\\0
	\end{array}\right),
	\end{equation}
	where $L_{s_1s_1}$ is the entry $s_1s_1$ of the Laplacian  $L^{[s_2]}$, $B_{s_1}$ is the row $s_1$ of this Laplacian without considering the element in the diagonal and $L_U$ are the rows and columns of the unseeded vertices. Since $y_{s_1}=P(s_1 \sim s_1)=1$, we obtain the following linear system of equations
	\[L_Uy_U=-B_{s_1},\]
	which is the same linear system that the Random Walker solves (\cite{Grady2006} section III.B, equation (10)). Therefore $P(q\sim s_1)=y_q=x_q^{s_1}$ for all $q$.
\end{proof}

\section{Power Watershed}
\label{app_power_watershed}
In this section we recall some definitions of \cite{Couprie2011}. Let  $G=(V,E,w)$ be an undirected edge-weighted graph and $s_1,s_2\in V$ two seeds as it has been considered in the main paper. In \cite{Couprie2011} the following objective function is proposed:
\begin{equation}
\label{eq_app:alg_Couprie}
x^*=\arg\min_x \sum_{e=(v,u)\in E}\left(w(e)\right)^{\alpha}\left(|x_v-x_u|\right)^{\beta},\ \text{s.t.} \  x_{s_1}=1, \ x_{s_2}=0.
\end{equation}
This objective generalizes a set of segmentation algorithms depending on the choice of parameters $\alpha$ and $\beta$. For instance, $\alpha=1$ and $\beta=2$ give the Random Walker's objective function.

The Power Watershed algorithm solves \eqref{eq_app:alg_Couprie} when  $\alpha\to\infty$. The algorithm is similar to Kruskal’s algorithm \cite{Kruskal1956}: a maximum weight spanning forest rooted in the seeds is computed iteratively, but at each plateau (maximal connected subgraphs with constant edge-weight) the following optimization problem is solved 
\begin{equation}
\label{eq_app:plateau_min}
\min_x\sum_{(u,v)\in E} |x_u-x_v|^\beta.
\end{equation}
 In case that $\beta=2$ this is equivalent to apply the Random Walker on the plateau.\\

\section{Proof of \thref{th:Pow_Wat=Max_Span_tree_ratio}}
\label{App:Proof_PWS_count_mSF}
\begin{theorem}[\thref{th:Pow_Wat=Max_Span_tree_ratio}]
	\thlabel{App_th:Pow_Wat=Max_Span_tree_ratio}
	Given two seeds $s_1$ and $s_2$, let us denote the potential of node $q$ being assigned to seed $s_1$ by the Power Watershed with $\beta=2$ as $x_{q}^{\text{PW}}$. Let further $ w_{\max}$ be $\max_{f\in \Forest[s_1][s_2]} w(f)$. Then
	\[x_{q}^{\text{PW}}=\frac{\left|\{f\in \Forest[s_1][s_2][q] \ : \  w(f)=w_{\max}\}\right|}{\left|\{f\in \Forest[s_1][s_2] \ : \  w(f)=w_{\max}\}\right|}\eqqcolon P_{\infty}(q\sim s_1).\]
\end{theorem}
\begin{proof}
	It has already been proven in Theorem 3 of \cite{Couprie2011} that the potential computed by the algorithm of the Power Watershed is equal to the limit of the Random Walker probabilities when the weights are raised to $\alpha\to\infty$. Since the Probabilistic Watershed probabilities are the same as the Random Walker probabilities (see Section 4), we just need to show that the limit of the Probabilistic Watershed with the weights raised to $\alpha\to\infty$ is counting MSFs.
	\begin{equation}
	\label{eq:Power_Watershed_MSF}
	\begin{aligned}
	&P_{\alpha}(q\sim s_1)\coloneqq	\frac{\displaystyle\sum_{f\in \Forest[s_1][s_2][q]}\prod_{e\in f}w(e)^{\alpha}}{\displaystyle\sum_{f\in \Forest[s_1][s_2]}\prod_{e\in f}w(e)^{\alpha}}=	\frac{\displaystyle\sum_{f\in \Forest[s_1][s_2][q]}w(f)^{\alpha}}{\displaystyle\sum_{f\in \Forest[s_1][s_2]}w(f)^{\alpha}}=\frac{\displaystyle\sum_{f\in \Forest[s_1][s_2][q]}\left(\frac{w(f)}{w_{\max}}\right)^{\alpha}}{\displaystyle\sum_{f\in \Forest[s_1][s_2]}\left(\frac{w(f)}{w_{\max}}\right)^{\alpha}}\xrightarrow[(\star)]{\alpha\to \infty} 
	 P_{\infty}(q\sim s_1).
	\end{aligned}
	\end{equation}
	In ($\star$) we used the fact that ${\frac{w(f)}{w_{\max}}<1 \iff w(f)\neq w_{\max} }$. When $\alpha\to\infty$, only for the MSFs the fraction $\left(w(f)/w_{\max}\right)^{\alpha}$  does not tend to 0, but to 1. Thus, we are counting MSFs.
\end{proof}

\tikzset{
	dot/.style 2 args={fill, circle, inner sep=0pt, label={#1:\scriptsize #2}},	
	fulldot/.style 2 args={circle,draw,minimum size=0.3cm,inner sep=0pt, label={#1:\scriptsize #2}},
	seed node/.style={circle,draw,minimum size=0.3cm,inner sep=0pt]},
	purria node/.style={circle,draw,minimum size=0.15cm,inner sep=0pt]}
}

\newcommand{\qsu}{$8\cdot 10^{-4}$}

\newcommand{\qvdosu}{$10^{-3}$}
\newcommand{\qvdosdos}{$10^{-3}$}

\newcommand{\qvtresu}{$10^{-2}$}
\newcommand{\qvtresdos}{$10^{-2}$}
\newcommand{\qvtrestres}{$10^{-2}$}

\newcommand{\qvquatreu}{$10^{-1}$}
\newcommand{\qvquatredos}{$10^{-1}$}
\newcommand{\qvquatretres}{$10^{-1}$}
\newcommand{\qvquatrequatre}{$10^{-1}$}

\newcommand{\svdosu}{0}
\newcommand{\svdosdos}{0}

\newcommand{\svtresu}{0}
\newcommand{\svtresdos}{0}
\newcommand{\svtrestres}{0}

\newcommand{\svquatreu}{0}
\newcommand{\svquatredos}{0}
\newcommand{\svquatretres}{0}
\newcommand{\svquatrequatre}{0}
\def \scale{1}
\def \scaled{0.8}
\def \widthtext{0.32}

\begin{figure*}[h]
	\centering
\begin{subfigure}{\widthtext\textwidth}
			\centering
			\begin{tikzpicture}[]

			\node[seed node,opacity=.5,black,text opacity=1] (q) at (0,0) {\tiny$q$};
			
			\node[seed node,opacity=.5,black,fill=cyan,text opacity=1] (s1) at (\scale*-1.7,0) {\tiny$s_1$};
			
			\node[seed node,opacity=.5,black,fill=green,text opacity=1] (s2) at (\scale*-1.0,\scale*1.5) {\tiny$s_2$};
			\node[purria node] (v21) at (\scale*-0.25,\scale*1) {};
			\node[purria node] (v22) at (\scale*-1,\scale*0.5) {};

			\node[seed node,opacity=.5,black,fill=yellow,text opacity=1] (s3) at (\scale*2,\scale*1) {\tiny$s_3$};
			\node[purria node] (v31) at (\scale*1,\scale*1.5) {};
			\node[purria node] (v32) at (\scale*1,\scale*0.75) {};
			\node[purria node] (v33) at (\scale*1,\scale*0) {};

			\node[seed node,opacity=.5,black,fill=red,text opacity=1] (s4) at (\scaled*0,\scaled*-2) {\tiny$s_4$};
			\node[purria node] (v41) at (\scaled*-1,\scaled*-1) {};
			\node[purria node] (v42) at (\scaled*-0.33,\scaled*-1) {};
			\node[purria node] (v43) at (\scaled*0.33,\scaled*-1) {};
			\node[purria node] (v44) at (\scaled*1,\scaled*-1) {};

			\path[-,draw]
			(q) edge node[below left] {\tiny \qsu} (s1)

			(q) edge node[left] {\tiny \qvdosu} (v21)
			(q) edge node[ ] {\tiny} (v22)

			(q) edge [bend left] node [pos=0.7, above] {\tiny \qvtresu } (v31)
			(q) edge node[right] {\tiny } (v32)
			(q) edge node[above right] {\tiny } (v33)

			(q) edge node[left] {\tiny } (v41)
			(q) edge node[below left] {\tiny } (v42)
			(q) edge node[below right] {\tiny } (v43)
			(q) edge node[ right] {\tiny \qvquatrequatre} (v44)
			
			(s2) edge node[above] {\tiny } (v21)
			(s2) edge node[right] {\tiny \svdosdos} (v22)

			(s3) edge node[above] {\tiny \svtresu} (v31)
			(s3) edge node[] {\tiny } (v32)
			(s3) edge node[below] {\tiny } (v33)
			
			(s4) edge node[above left] {\tiny } (v41)
			(s4) edge node[below left] {\tiny } (v42)
			(s4) edge node[below left] {\tiny } (v43)
			(s4) edge node[below ] {\tiny \svquatrequatre} (v44);

			\end{tikzpicture}
			\caption{Graph with four seeds}
			\label{sfig1:mu_behaviour}
		\end{subfigure}
		\begin{subfigure}{\widthtext\textwidth}
			\centering
			\begin{tikzpicture}[]

			\node[seed node,opacity=.5,black,fill=red,text opacity=1] (q) at (0,0) {\tiny$q$};
			
			\node[seed node,opacity=.5,black,fill=cyan,text opacity=1] (s1) at (\scale*-1.7,0) {\tiny$s_1$};
			
			\node[seed node,opacity=.5,black,fill=green,text opacity=1] (s2) at (\scale*-1.0,\scale*1.5) {\tiny$s_2$};
			\node[purria node,opacity=.5,black,fill=green,text opacity=1] (v21) at (\scale*-0.25,\scale*1) {};
			\node[purria node,opacity=.5,black,fill=green,text opacity=1] (v22) at (\scale*-1,\scale*0.5) {};

			\node[seed node,opacity=.5,black,fill=yellow,text opacity=1] (s3) at (\scale*2,\scale*1) {\tiny$s_3$};
			\node[purria node,opacity=.5,black,fill=yellow,text opacity=1] (v31) at (\scale*1,\scale*1.5) {};
			\node[purria node,opacity=.5,black,fill=yellow,text opacity=1] (v32) at (\scale*1,\scale*0.75) {};
			\node[purria node,opacity=.5,black,fill=yellow,text opacity=1] (v33) at (\scale*1,\scale*0) {};

			\node[seed node,opacity=.5,black,fill=red,text opacity=1] (s4) at (\scaled*0,\scaled*-2) {\tiny$s_4$};
			\node[purria node,opacity=.5,black,fill=red,text opacity=1] (v41) at (\scaled*-1,\scaled*-1) {};
			\node[purria node,opacity=.5,black,fill=red,text opacity=1] (v42) at (\scaled*-0.33,\scaled*-1) {};
			\node[purria node,opacity=.5,black,fill=red,text opacity=1] (v43) at (\scaled*0.33,\scaled*-1) {};
			\node[purria node,opacity=.5,black,fill=red,text opacity=1] (v44) at (\scaled*1,\scaled*-1) {};
			
			\path[-,draw]
			(q) edge node[below left] {} (s1)

			(q) edge node[left] {} (v21)
			(q) edge node[left] {} (v22)
			
			(q) edge [bend left] node[above left] {} (v31)
			(q) edge node[above right] {} (v32)
			(q) edge node[above right] {} (v33)
			
			(q) edge node[left] {} (v41)
			(q) edge node[below left] {} (v42)
			(q) edge node[below right] {} (v43)
			(q) edge node[ right] {} (v44)

			(s2) edge node[below] {} (v21)
			(s2) edge node[below] {} (v22)

			(s3) edge node[below] {} (v31)
			(s3) edge node[below] {} (v32)
			(s3) edge node[below] {} (v33)

			(s4) edge node[below left] {} (v41)
			(s4) edge node[below left] {} (v42)
			(s4) edge node[below left] {} (v43)
			(s4) edge node[below left] {} (v44);

			\end{tikzpicture}
			\caption{$\mu=0$}
			\label{sfig2:mu_behaviour}
		\end{subfigure}
		\begin{subfigure}{\widthtext\textwidth}
			\centering
			\begin{tikzpicture}[]

			\node[seed node,opacity=.5,black,fill=yellow,text opacity=1] (q) at (0,0) {\tiny$q$};
			
			\node[seed node,opacity=.5,black,fill=cyan,text opacity=1] (s1) at (\scale*-1.7,0) {\tiny$s_1$};
			
			\node[seed node,opacity=.5,black,fill=green,text opacity=1] (s2) at (\scale*-1.0,\scale*1.5) {\tiny$s_2$};
			\node[purria node,opacity=.5,black,fill=green,text opacity=1] (v21) at (\scale*-0.25,\scale*1) {};
			\node[purria node,opacity=.5,black,fill=green,text opacity=1] (v22) at (\scale*-1,\scale*0.5) {};

			\node[seed node,opacity=.5,black,fill=yellow,text opacity=1] (s3) at (\scale*2,\scale*1) {\tiny$s_3$};
			\node[purria node,opacity=.5,black,fill=yellow,text opacity=1] (v31) at (\scale*1,\scale*1.5) {};
			\node[purria node,opacity=.5,black,fill=yellow,text opacity=1] (v32) at (\scale*1,\scale*0.75) {};
			\node[purria node,opacity=.5,black,fill=yellow,text opacity=1] (v33) at (\scale*1,\scale*0) {};

			\node[seed node,opacity=.5,black,fill=red,text opacity=1] (s4) at (\scaled*0,\scaled*-2) {\tiny$s_4$};
			\node[purria node,opacity=.5,black,fill=red,text opacity=1] (v41) at (\scaled*-1,\scaled*-1) {};
			\node[purria node,opacity=.5,black,fill=red,text opacity=1] (v42) at (\scaled*-0.33,\scaled*-1) {};
			\node[purria node,opacity=.5,black,fill=red,text opacity=1] (v43) at (\scaled*0.33,\scaled*-1) {};
			\node[purria node,opacity=.5,black,fill=red,text opacity=1] (v44) at (\scaled*1,\scaled*-1) {};

			\path[-,draw]
			(q) edge node[below left] {} (s1)

			(q) edge node[left] {} (v21)
			(q) edge node[left] {} (v22)
			
			(q) edge [bend left] node[above left] {} (v31)
			(q) edge node[above right] {} (v32)
			(q) edge node[above right] {} (v33)
			
			(q) edge node[left] {} (v41)
			(q) edge node[below left] {} (v42)
			(q) edge node[below right] {} (v43)
			(q) edge node[ right] {} (v44)

			(s2) edge node[below] {} (v21)
			(s2) edge node[below] {} (v22)

			(s3) edge node[below] {} (v31)
			(s3) edge node[below] {} (v32)
			(s3) edge node[below] {} (v33)

			(s4) edge node[below left] {} (v41)
			(s4) edge node[below left] {} (v42)
			(s4) edge node[below left] {} (v43)
			(s4) edge node[below left] {} (v44);

			\end{tikzpicture}
			\caption{$\mu=10$}
			\label{sfig3:mu_behaviour}
		\end{subfigure}
		\begin{subfigure}{\widthtext\textwidth}
			\centering
			\begin{tikzpicture}[]

			\node[seed node,opacity=.5,black,fill=green,text opacity=1] (q) at (0,0) {\tiny$q$};
			
			\node[seed node,opacity=.5,black,fill=cyan,text opacity=1] (s1) at (\scale*-1.7,0) {\tiny$s_1$};
			
			\node[seed node,opacity=.5,black,fill=green,text opacity=1] (s2) at (\scale*-1.0,\scale*1.5) {\tiny$s_2$};
			\node[purria node,opacity=.5,black,fill=green,text opacity=1] (v21) at (\scale*-0.25,\scale*1) {};
			\node[purria node,opacity=.5,black,fill=green,text opacity=1] (v22) at (\scale*-1,\scale*0.5) {};

			\node[seed node,opacity=.5,black,fill=yellow,text opacity=1] (s3) at (\scale*2,\scale*1) {\tiny$s_3$};
			\node[purria node,opacity=.5,black,fill=yellow,text opacity=1] (v31) at (\scale*1,\scale*1.5) {};
			\node[purria node,opacity=.5,black,fill=yellow,text opacity=1] (v32) at (\scale*1,\scale*0.75) {};
			\node[purria node,opacity=.5,black,fill=yellow,text opacity=1] (v33) at (\scale*1,\scale*0) {};

			\node[seed node,opacity=.5,black,fill=red,text opacity=1] (s4) at (\scaled*0,\scaled*-2) {\tiny$s_4$};
			\node[purria node,opacity=.5,black,fill=red,text opacity=1] (v41) at (\scaled*-1,\scaled*-1) {};
			\node[purria node,opacity=.5,black,fill=red,text opacity=1] (v42) at (\scaled*-0.33,\scaled*-1) {};
			\node[purria node,opacity=.5,black,fill=red,text opacity=1] (v43) at (\scaled*0.33,\scaled*-1) {};
			\node[purria node,opacity=.5,black,fill=red,text opacity=1] (v44) at (\scaled*1,\scaled*-1) {};

			\path[-,draw]
			(q) edge node[below left] {} (s1)

			(q) edge node[left] {} (v21)
			(q) edge node[left] {} (v22)
			
			(q) edge [bend left] node[above left] {} (v31)
			(q) edge node[above right] {} (v32)
			(q) edge node[above right] {} (v33)
			
			(q) edge node[left] {} (v41)
			(q) edge node[below left] {} (v42)
			(q) edge node[below right] {} (v43)
			(q) edge node[ right] {} (v44)

			(s2) edge node[below] {} (v21)
			(s2) edge node[below] {} (v22)

			(s3) edge node[below] {} (v31)
			(s3) edge node[below] {} (v32)
			(s3) edge node[below] {} (v33)

			(s4) edge node[below left] {} (v41)
			(s4) edge node[below left] {} (v42)
			(s4) edge node[below left] {} (v43)
			(s4) edge node[below left] {} (v44);

			\end{tikzpicture}
			\caption{$\mu=100$}
			\label{sfig4:mu_behaviour}
		\end{subfigure}
		\begin{subfigure}{\widthtext\textwidth}
			\centering
			\begin{tikzpicture}[]

			\node[seed node,opacity=.5,black,fill=cyan,text opacity=1] (q) at (0,0) {\tiny$q$};
			
			\node[seed node,opacity=.5,black,fill=cyan,text opacity=1] (s1) at (\scale*-1.7,0) {\tiny$s_1$};
			
			\node[seed node,opacity=.5,black,fill=green,text opacity=1] (s2) at (\scale*-1.0,\scale*1.5) {\tiny$s_2$};
			\node[purria node,opacity=.5,black,fill=green,text opacity=1] (v21) at (\scale*-0.25,\scale*1) {};
			\node[purria node,opacity=.5,black,fill=green,text opacity=1] (v22) at (\scale*-1,\scale*0.5) {};

			\node[seed node,opacity=.5,black,fill=yellow,text opacity=1] (s3) at (\scale*2,\scale*1) {\tiny$s_3$};
			\node[purria node,opacity=.5,black,fill=yellow,text opacity=1] (v31) at (\scale*1,\scale*1.5) {};
			\node[purria node,opacity=.5,black,fill=yellow,text opacity=1] (v32) at (\scale*1,\scale*0.75) {};
			\node[purria node,opacity=.5,black,fill=yellow,text opacity=1] (v33) at (\scale*1,\scale*0) {};

			\node[seed node,opacity=.5,black,fill=red,text opacity=1] (s4) at (\scaled*0,\scaled*-2) {\tiny$s_4$};
			\node[purria node,opacity=.5,black,fill=red,text opacity=1] (v41) at (\scaled*-1,\scaled*-1) {};
			\node[purria node,opacity=.5,black,fill=red,text opacity=1] (v42) at (\scaled*-0.33,\scaled*-1) {};
			\node[purria node,opacity=.5,black,fill=red,text opacity=1] (v43) at (\scaled*0.33,\scaled*-1) {};
			\node[purria node,opacity=.5,black,fill=red,text opacity=1] (v44) at (\scaled*1,\scaled*-1) {};

			\path[-,draw]
			(q) edge node[below left] {} (s1)

			(q) edge node[left] {} (v21)
			(q) edge node[left] {} (v22)
			
			(q) edge [bend left] node[above left] {} (v31)
			(q) edge node[above right] {} (v32)
			(q) edge node[above right] {} (v33)
			
			(q) edge node[left] {} (v41)
			(q) edge node[below left] {} (v42)
			(q) edge node[below right] {} (v43)
			(q) edge node[ right] {} (v44)

			(s2) edge node[below] {} (v21)
			(s2) edge node[below] {} (v22)

			(s3) edge node[below] {} (v31)
			(s3) edge node[below] {} (v32)
			(s3) edge node[below] {} (v33)

			(s4) edge node[below left] {} (v41)
			(s4) edge node[below left] {} (v42)
			(s4) edge node[below left] {} (v43)
			(s4) edge node[below left] {} (v44);

			\end{tikzpicture}
			\caption{$\mu\to\infty$, Watershed}
			
			\label{sfig5:mu_behaviour}
	\end{subfigure}
	\caption{Effect of the inverse temperature $\mu$ on Probabilistic Watershed solutions. (\ref{sfig1:mu_behaviour}) shows a graph with 4 seeds and edge costs, $c(e)$. All paths from the query node $q$ to a seed $s_i$ have the same cost (only indicated once per seed). (\ref{sfig2:mu_behaviour}) - (\ref{sfig4:mu_behaviour}) show the Probabilistic Watershed's segmentation for edge weights $\exp(-\mu \, c(e))$.  As $\mu$ grows, $q$'s assignment changes from a weight-independent (maximum entropy) one over two Random Walker assignments to the Watershed assignment (lowest entropy).}
	\label{fig:mu_behaviour}	
\end{figure*}
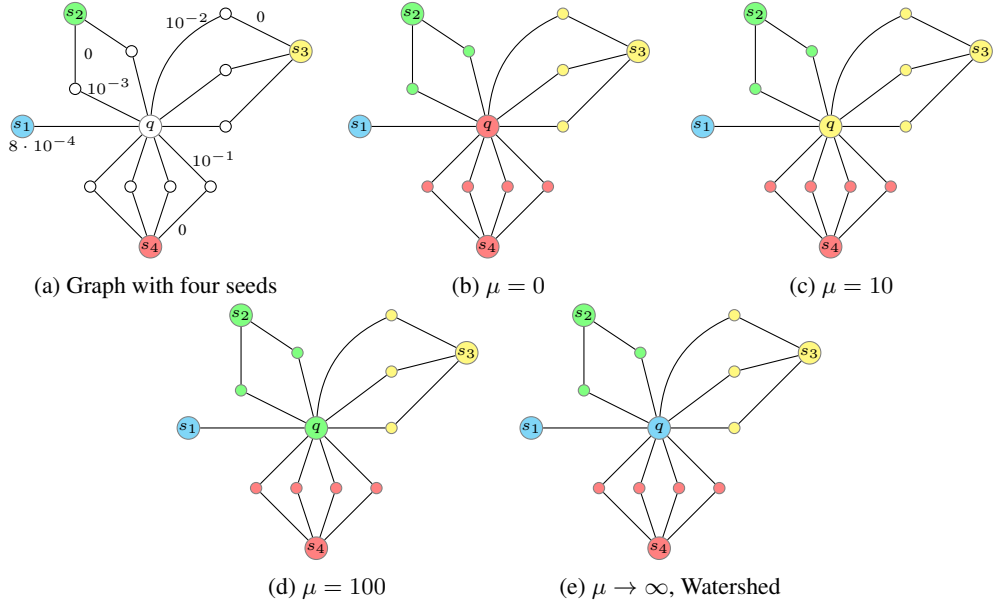

 \section{Effect of the entropy on the Probabilistic Watershed}
 \label{app_mu_behavior}
 \figurename{} \ref{fig:mu_behaviour} illustrates how the forest distribution's entropy interpolates between (Power) Watershed and Probabilistic Watershed / Random Walker with decreasing sensitivity to edge-costs.

\section{Edge and node probabilities in the Power Watershed}
\label{App:Edge_prob}

\begin{figure}[h]
\captionsetup[subfigure]{justification=centering}
    \begin{subfigure}{0.49\textwidth}
    \centering
    \includegraphics[height=5cm]{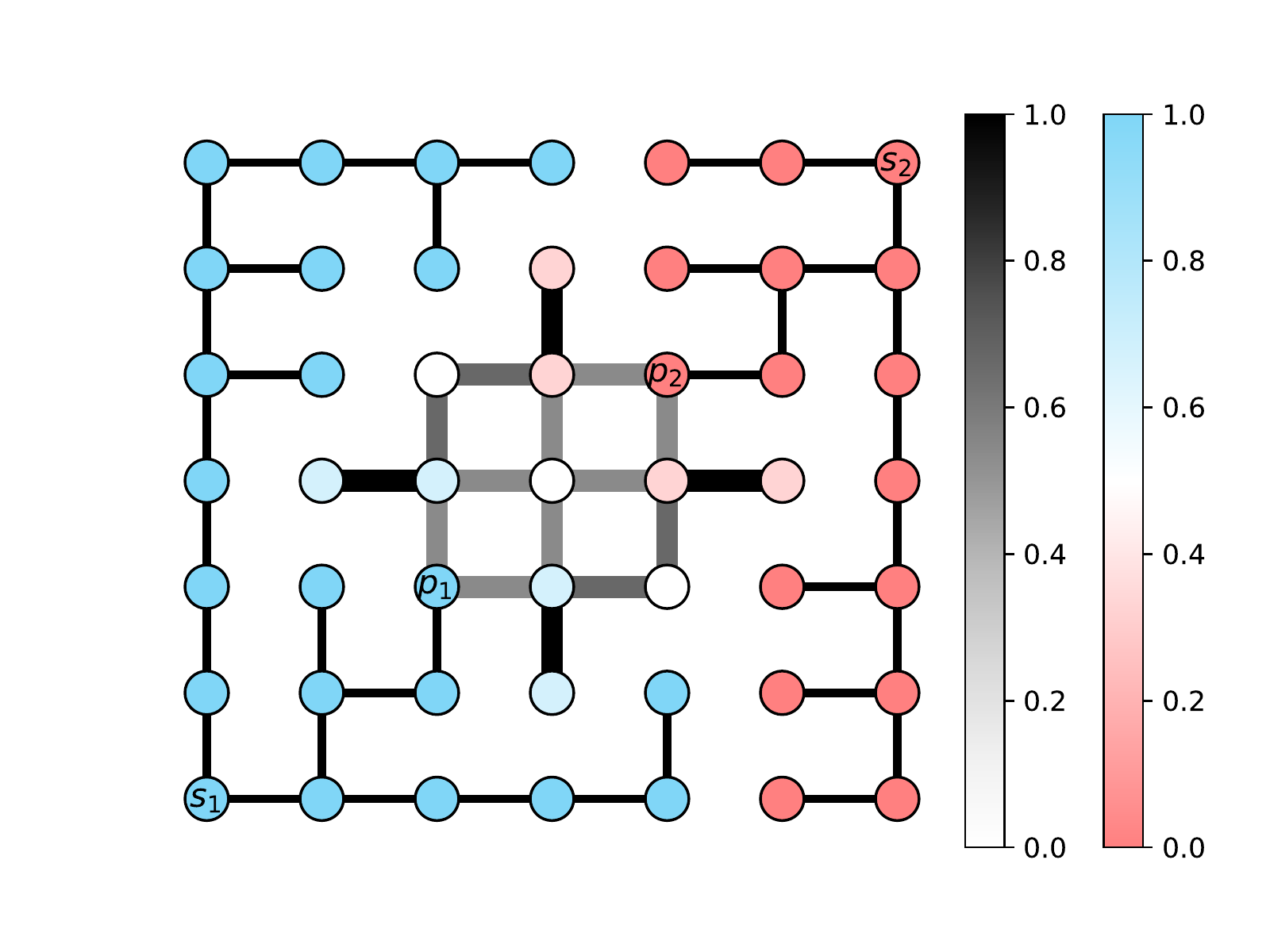}
    \caption{$P(\text{node} \sim s_1)$ and \\ $P(\text{edge} \in \text{some mSF}$)}
    \label{App_sfig:edge_present}
    \end{subfigure}
    \begin{subfigure}{0.49\textwidth}
    \centering
    \includegraphics[height=5cm]{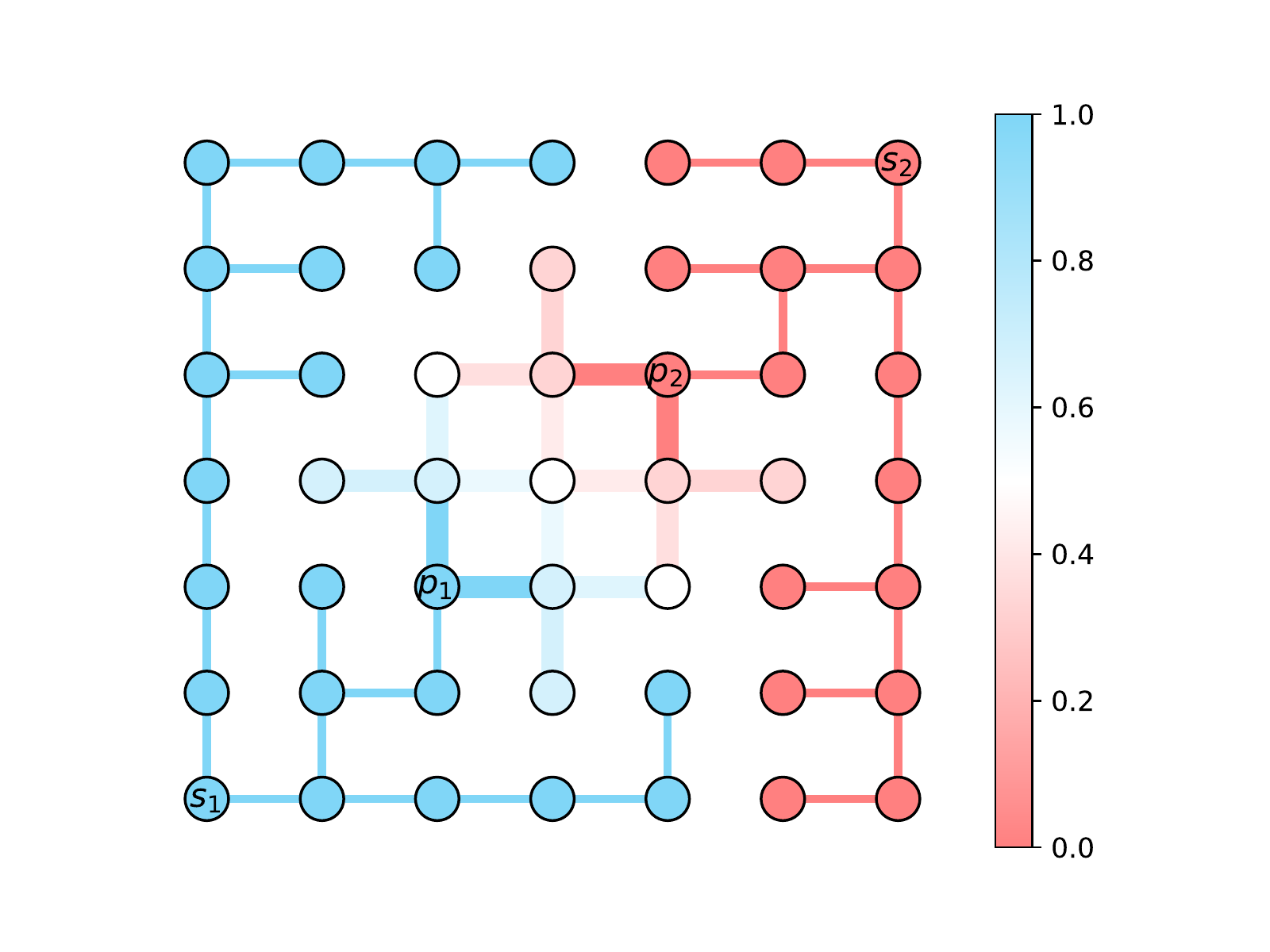}
    \caption{$P(\text{node} \sim s_1)$ and \\ $P(\text{edge} \sim s_1 | \text{edge} \in \text{some msF})$}
    \label{App_sfig:edge_given_present}
    \end{subfigure}
    \begin{subfigure}{0.49\textwidth}
    \centering
    \includegraphics[height=5cm]{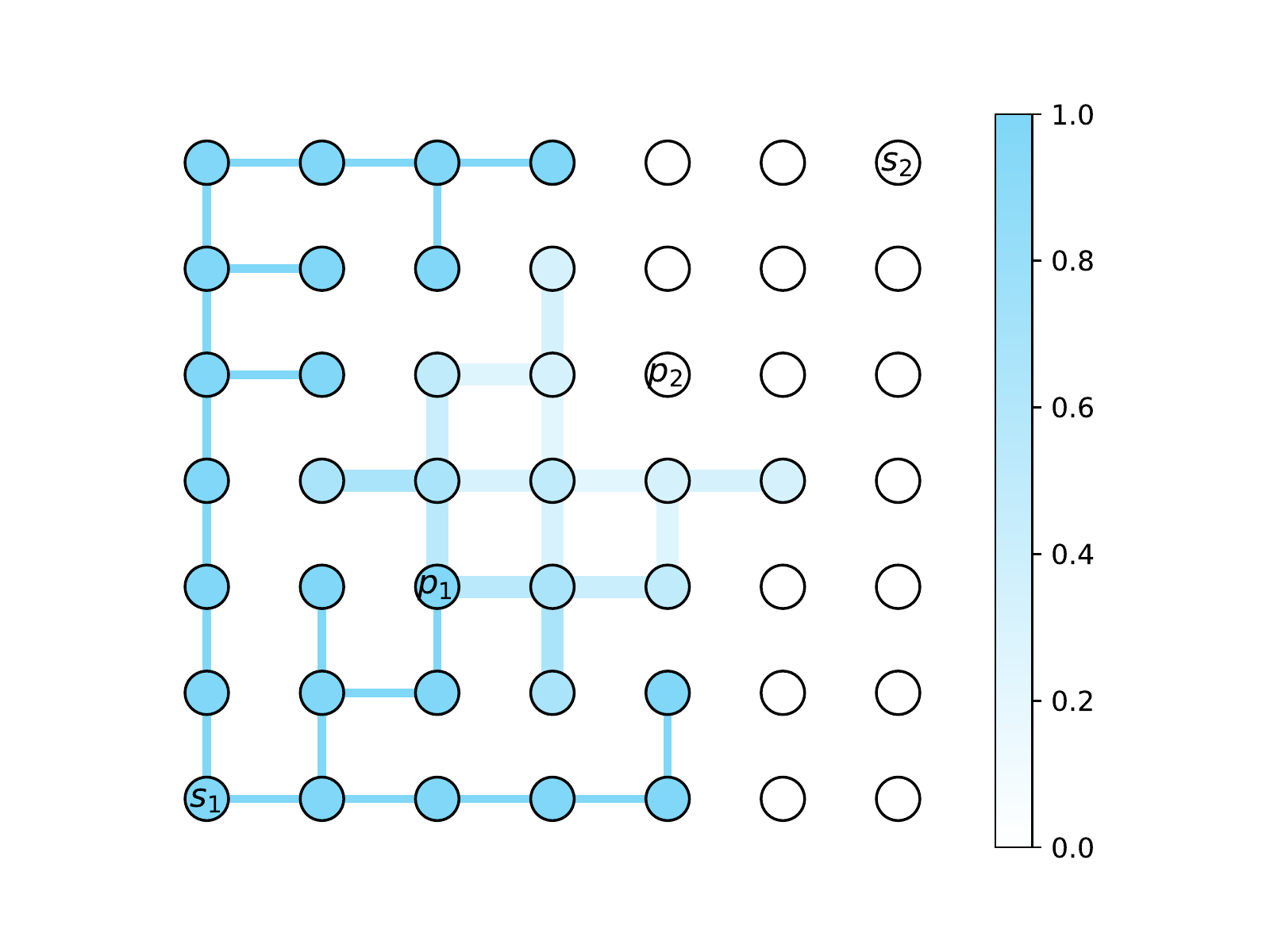}
    \caption{$P(\text{node} \sim s_1)$ and\\
    $P(\text{edge} \sim s_1, \text{edge} \in \text{some mSF})$}
    \label{App_sfig:edge_s_1}
    \end{subfigure}
    \begin{subfigure}{0.49\textwidth}
    \centering
    \includegraphics[height=5cm]{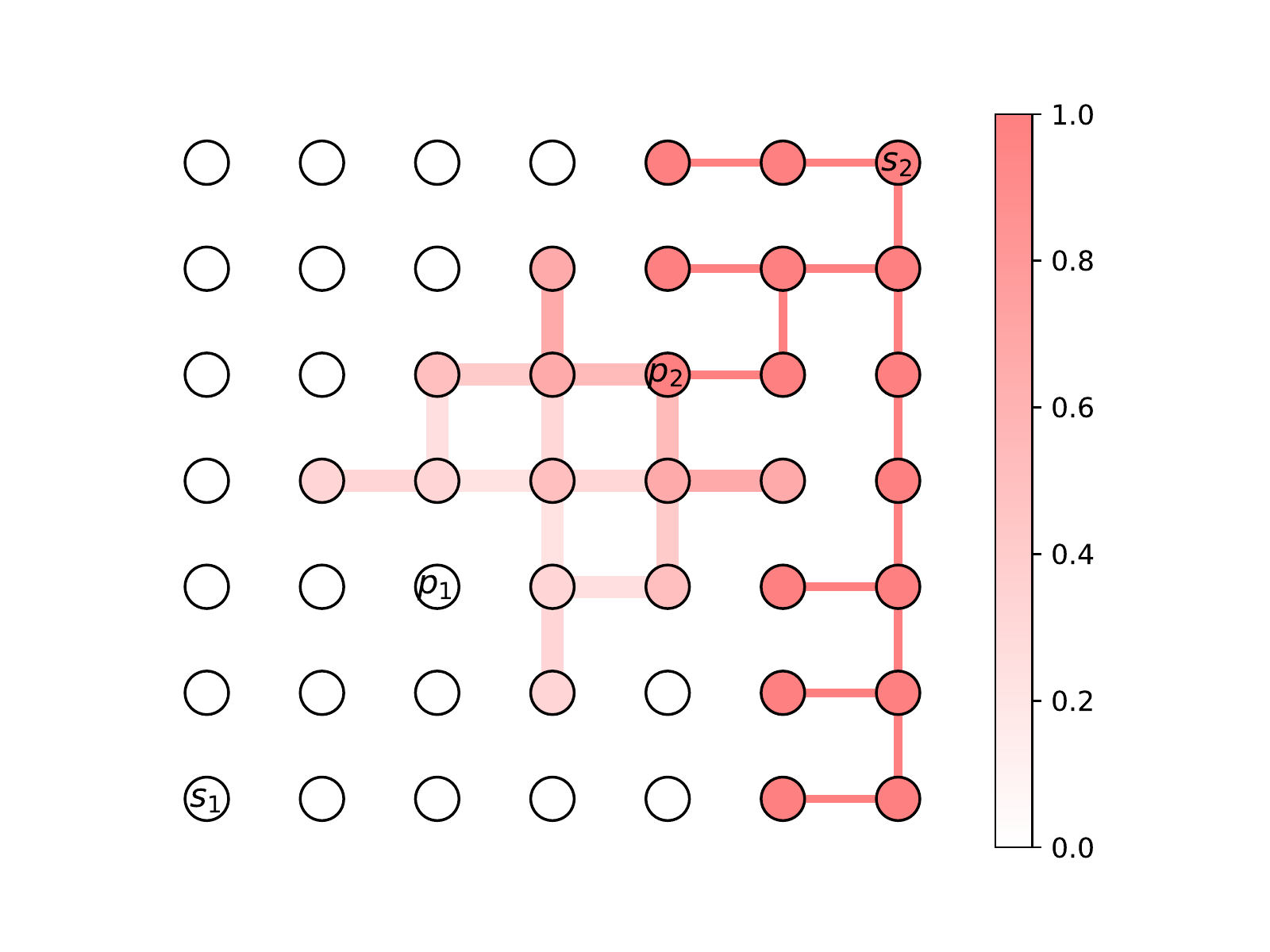}
    \caption{$P(\text{node} \sim s_2)$ and\\
    $P(\text{edge} \sim s_2, \text{edge} \in \text{some mSF})$}
    \label{App_sfig:edge_s_2}
    \end{subfigure}
    \caption{Power Watershed result on a grid graph with seeds $s_1$, $s_2$ and with random edge-costs outside a plateau of edges with the same cost (wide edges). By the results in Theorem 5.1, the Power Watershed counts mSFs. This is illustrated with both the node- and edge-colors. (\ref{App_sfig:edge_present}-\ref{App_sfig:edge_s_2}) The nodes are colored by their probability of belonging to seed $s_1$ ($s_2$), i.e. by the share of mSFs that connect a given node to $s_1$ ($s_2)$. (\ref{App_sfig:edge_present}) The edge-color indicates the share of mSFs in which the edge is present. (\ref{App_sfig:edge_given_present}) The edge-color indicates the share of mSFs in which the edge is connected to seed $s_1$ among the mSFs that contain the edge. (\ref{App_sfig:edge_s_1} - \ref{App_sfig:edge_s_2}) The edge-color indicates the share of mSFs in which the edge is connected to $s_1$ or $s_2$, respectively, among all mSFs. 
    }
    \label{App_fig:edge_presence}
\end{figure}

In this chapter, we elaborate the minimum spanning forest (mSF) counting interpretation of the Power Watershed. \figurename{} \ref{App_fig:edge_presence} shows a graph $G$ with a single plateau $P$, a maximal connected subgraph of constant edge-cost $c$. To simplify our exposition, we made sure that there is exactly one path with maximum cost below $c$ from each seed to $P$. The nodes at the end of these paths are called $p_1$ and $p_2$, respectively. We illustrate the mSF-counting nature of the Power Watershed both on nodes and on edges.\\
In \figurename{} \ref{App_sfig:edge_present}, we show the probability of an edge being present in a mSF. Outside the plateau, the edges are either part of every or of no mSF. All mSFs agree on these edges. They can be found by a variant of Kruskal's greedy algorithm which iteratively adds edges of minimal cost, while avoiding cycles and connections between the two seeds. Therefore, the edges outside the plateau are only black or white in \figurename{} \ref{App_sfig:edge_present}. On the plateau all spanning forests have the same, minimal cost. Here, Power Watershed performs the Random Walker, or - in our forest-framework - counts spanning forests. Therefore, the edges on $P$ typically have a probability of being present in a mSF strictly between $0$ and $1$. Note that the final segmentation can be read-off from the edge probabilities in \figurename{} \ref{App_sfig:edge_present} outside the plateau (as in each mSF in our example every node outside the plateau can be reached from a seed without entering the plateau) but not on the plateau without the node potentials.\\
In \figurename{}s \ref{App_sfig:edge_given_present}-\ref{App_sfig:edge_s_2}, we show how likely an edge is connected to either of the seeds in a mSF. Again, all the edges outside the plateau are either always connected to the same seed in all mSFs or never part of any mSF. In the latter case, the conditional probability in \ref{App_sfig:edge_given_present} is not defined; we colored them white, which corresponds to the uninformed probability of $0.5$ for ease of presentation. The closer an edge of $P$ is to the node $p_1$, where the subtree of $s_1$ connects to the plateau, the higher its probability to be connected to $s_1$ among the mSFs that contain this edge (\figurename{} \ref{App_sfig:edge_given_present}) and also among all mSFs (\figurename{} \ref{App_sfig:edge_s_1}). The same holds for $s_2$ in \figurename{}s \ref{App_sfig:edge_given_present} and \ref{App_sfig:edge_s_2}. Note that in both \figurename{} \ref{App_sfig:edge_s_1} and \figurename{} \ref{App_sfig:edge_s_2} the color intensity of every edge $e = \{u,v\}$ is at most as high as that of $u$ or $v$. This is because whenever $e$ is connected to some seed in a mSF $f$, both $u$ and $v$ are connected to that seed in 
$f$, too.\\
We computed the probability of an edge being present in a spanning forest on the plateau by the generalization of the MTT in Lemma 1.9 of \cite{Shayan_spanning_tree}, see also Theorem 2 of \cite{Teixeira2013} for a version on unweighted graphs. Then for each edge $e = \{u, v\}$ on $P$, we merged $u$ and $v$ into a new node $q_e$, thus obtaining a minor $P_e$ of $P$. On $P_e$ in turn, we computed the Probabilistic Watershed probabilities $P_{P_e}(p_1 \sim q_e)$, hence finding the share of 2-forests in $P_e$ isolating $p_1$ and $p_2$ that connect $q_e$ to $p_1$. This is nothing but the share of 2-forests in $P$ separating $p_1$ and $p_2$ that contain $e$ and connect it to $p_1$ among the 2-forests separating $p_1$ and $p_2$ that contain $e$. Multiplying this with the probability that an edge is part of any mSF gives the share of mSFs in $G$, which contain $e$ and connected it to $s_1$, among all mSFs separating $s_1$ and $s_2$.

\section{Rough lower bound for the number of forests in a grid graph}
\label{App:Bound_Forests}
In this chapter we derive a rough lower bound on the number of spanning forests that separate $k$ given seeds in a two-dimensional grid graph. We refer to these forests as ``$k$-forests". If there is some $n \times m$ subgrid without any seeds then the number of $k$-forests is at least as large as the number of spanning trees in the subgrid. This is because there are $k$-forests in which all nodes in the subgrid belong the tree of some of the seeds. We can compute the number of spanning trees $N_T$ in a grid graph with $n$ rows and $m$ columns by the closed-form formula (see \cite{tzeng2000spanning} Theorem 1):
\begin{align}
    N_T(n,m) = \frac{2^{nm -1}}{nm} \cdot \prod_{\substack{i=0, \dots, n-1,\\ j=0,\dots, m-1,\\ (i,j) \neq (0,0)}} \left(2 - \cos\left(\frac{i\pi}{n}\right) - \cos\left(\frac{j \pi}{m}\right)\right)
\end{align}
The image in \figurename{} 2 of the main paper has a seed-free part of size $87 \times 272$, see \figurename{} \ref{fig:shaded_region} below. This yields the following lower bound for the number of $13$-forests separating the $13$ seeds:
\begin{gather}
    N_T(87,272) \approx 10^{11847}
\end{gather}

\begin{figure}[h]
    \centering
    \includegraphics[width = 0.8\textwidth]{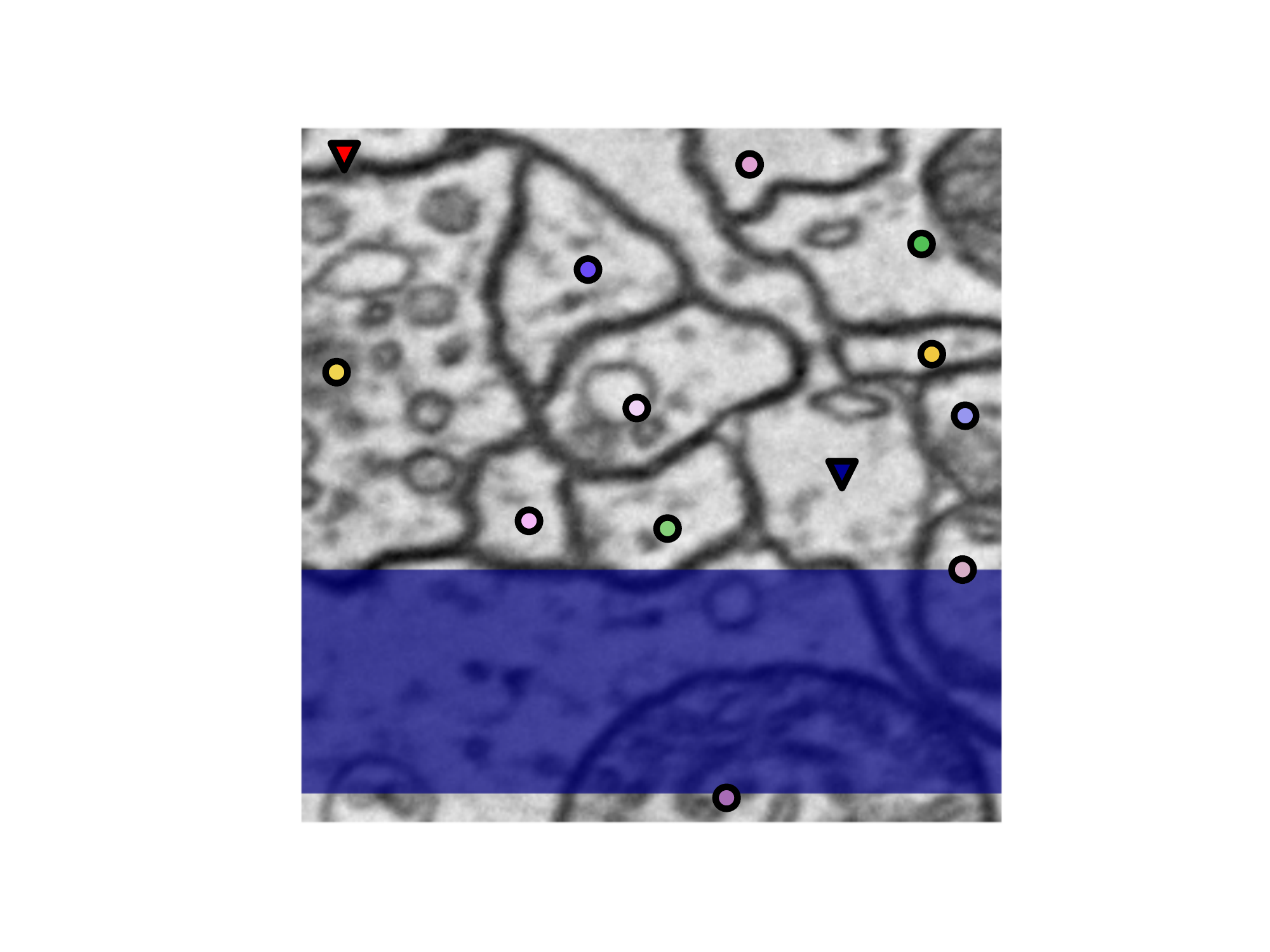}
    \caption{The shaded region was used to obtain a rough lower bound on the number of forests separating the seeds. There are about $10^{11847}$ spanning trees in the grid graph that corresponds to the shaded region and hence at least as many forests in the whole graph which separate the seeds.}
    \label{fig:shaded_region}
\end{figure}

\end{document}